\pdfoutput=1
\documentclass[runningheads]{llncs}
\usepackage{amsmath,amssymb}
\usepackage{graphicx}
\usepackage{booktabs,multirow,array}
\usepackage[table]{xcolor}
\usepackage[pdfusetitle,colorlinks=true,linkcolor={black!15!blue},
            citecolor={black!35!green},urlcolor={black!25!blue}]{hyperref}
\usepackage{algorithm,algorithmic}
\usepackage{microtype}
\definecolor{fairisgreen}{HTML}{E8F5E9}
\newcommand{\fairis}{\textsc{Fairis}}
\newcommand{\fairisn}{\textsc{Fairis-N}}
\newcommand{\fairfed}{\textsc{FairFed}}
\newcommand{\fedavg}{\textsc{FedAvg}}
\newcommand{\std}[1]{{\small$\pm$#1}}
\newcommand{\Fk}{\mathcal{F}_k}
\newcommand{\Fg}{\mathcal{F}_{\mathrm{global}}}
\newcommand{\wbar}{\bar{\omega}}
\newcommand{\PPT}{\mathsf{PPT}}
\begin{document}
\title{\texorpdfstring{Fairis: Fairness-Aware Aggregation\\ with Provable Influence Containment\\ against Fairness Poisoning Attacks\\ in Collaborative Machine Learning}{Fairis: Fairness-Aware Aggregation with Provable Influence Containment against Fairness Poisoning Attacks}}
\author{Devharsh Trivedi\inst{1} \and
Nesrine Kaaniche\inst{2} \and
Nikos Triandopoulos\inst{3} \and
Maryline Laurent\inst{2} \and
Jackson Walters\inst{4}}
\titlerunning{Fairis: Fairness-Aware Aggregation with Influence Containment}
\authorrunning{D. Trivedi et al.}
\institute{Department of Computer Science, Bowie State University, MD, USA\\
\email{dtrivedi@bowiestate.edu}
\and SAMOVAR, T\'{e}l\'{e}com SudParis, Institut Polytechnique de Paris, Palaiseau, France
\and Department of Computer Science, Brown University, RI, USA
\and Northern Virginia Community College, VA, USA}
\maketitle
\begin{abstract}
Collaborative machine learning (ML) among financial institutions must be both group-fair and robust against deliberate adversarial manipulation. Existing fairness-aware aggregation methods remain formally vulnerable to \emph{fairness poisoning}: a malicious client that maximizes group disparity while preserving accuracy evades accuracy-based Byzantine defenses, and, in the threat model we study, \fairfed{}'s gap-based weighting can be gamed by an adversary who observes the global fairness score. We present \fairis{}, a server-side aggregation reweighting scheme in which each client's update receives the normalized weight $\omega_k=\wbar_k/\sum_j\wbar_j$ built from the unnormalized score $\wbar_k = \eta - \mathcal{F}_k$, with $\mathcal{F}_k\in[0,1]$ the local Equal Opportunity Difference (EOD) and $\eta>1$ a tunable security parameter. We prove three security properties, Monotone Weight Reduction (MWR), Demographic Participation (DP), and Non-Gamesmanship (NG), extend MWR to colluding minority coalitions, and show that combining MWR with server-side norm clipping bounds the adversary's displacement of the global model by $\omega_0 C$, strictly decreasing in its own reported disparity. Assuming clients report their fairness scores honestly, an assumption this paper does not discharge, \fairis{} is the only rule evaluated that \emph{guarantees} every client strictly positive weight while provably reducing an adversary's weight monotonically in its bias; clipped \fairfed{} sometimes reaches a lower weight but guarantees nothing and zeroes a client outright on Taiwan Credit. Against an adversary stealthy enough to evade accuracy-based defenses, staying within $0.04$ accuracy of benign, \fairis{} cuts its aggregation weight by $41$ to $54\,\%$ below a size-blind control on Taiwan. On routine non-IID partitions no aggregation rule dominates, and a uniform-weighting ablation shows that containment tracks how far the adversary's score separates from the honest mean, providing none when the honest population is already unfair.
\end{abstract}
\section{Introduction}\label{sec:intro}
Financial institutions each observe only a partial view of borrowers, and data-protection regulations and commercial confidentiality prevent sharing raw records. Yet, jointly trained models could improve default prediction, pricing fairness, and compliance for all parties. Collaborative and Federated Learning (FL)~\cite{mcmahan2017,gupta2018,thapa2022} keep data local and exchange only model parameters or activations. Two distinct threats arise in such deployments, and this paper addresses both.

The first threat is \emph{unintentional bias}: heterogeneous, non-independent and identically distributed (non-IID) client data amplifies group disparity in the shared representation, harming individuals and exposing institutions to regulatory liability under frameworks such as the EU AI Act~\cite{euaiact2024}. A growing body of work proposes fairness-aware aggregation rules to mitigate it~\cite{ezzeldin2023,meerza2024glocalfair,zhang2025sffl}, yet these methods were not designed to withstand adversarial clients who deliberately inject bias into the aggregation.

The second threat is \emph{deliberate fairness poisoning}: a malicious client trains a model that maximizes group disparity while preserving accuracy, evading accuracy-based Byzantine defenses. Fairness poisoning is a targeted subclass of Byzantine behavior that stays within normal accuracy bands; we do not claim robustness to arbitrary Byzantine updates, which standard distance-based defenses target (Appendix~\ref{app:background}). EAB-FL~\cite{meerza2024eabfl} injects bias through redundant parameter space while maintaining accuracy, and Kasyap et al.~\cite{kasyap2025} show fairness-constrained optimization attacks that raise demographic parity disparity by up to 90\,\% with a single malicious client, and report that \fairfed{} fails to contain them, with disparity rising $16.7\,\%$ under random partitioning and $46.4\,\%$ under attribute-based partitioning on Adult. To our knowledge, prior fairness-aware aggregation methods have not explicitly analyzed resistance to this threat class.

We present \fairis{}, which addresses both threats with one mechanism: weighting each client by $\eta-\mathcal{F}_k$, where $\mathcal{F}_k\in[0,1]$ is the absolute local Equal Opportunity Difference (EOD; $0$ fair, $1$ maximally biased) and $\eta>1$ is a tunable security parameter in the recommended range $(1,(K{+}1)/K]$ (values near $1$ penalize bias sharply, larger values approach uniform weighting). Because the weight uses the absolute score rather than its gap from a global average, a client cannot game it by tracking the global mean, and because $\eta>1$, every client keeps strictly positive weight, preserving subgroup representation. This containment of aggregation influence holds under the threat model of Section~\ref{sec:threat}, and in particular under Assumption~(A2) that clients report their fairness scores honestly; a client that forges a low score is out of scope: ruling this out needs a proof that the score was computed correctly, which authentication of the sender does not provide (Section~\ref{sec:limitations}).

\textbf{Contributions.} (1)~We formally define the fairness poisoning adversary, a security game, and three security properties: Monotone Weight Reduction (MWR), Demographic Participation (DP), and Non-Gamesmanship (NG) (Section~\ref{sec:threat}). (2)~We propose \fairis{}, an architecture-agnostic aggregation scheme with a tunable security parameter $\eta>1$ that interpolates between maximum fairness emphasis and equal participation; the weighting formula applies wherever a server aggregates client parameters, without modifying client training (Section~\ref{sec:fairis}). (3)~We prove all three properties hold under \fairis{} and prove that \fairfed{}'s gap-based rule violates NG (Section~\ref{sec:security}); these properties bound the attacker's \emph{aggregation weight} as a function of bias and, as a new result, extend to colluding minority coalitions (Corollary~\ref{cor:coalition}). They do not directly bind downstream honest-client EOD, which additionally requires fair clients to carry a sufficient fairness signal. (4)~We demonstrate empirically on three datasets (including the Adult Income benchmark used by \fairfed{} and EAB-FL) across both FL and SplitML architectures, that \fairis{} is competitive with baselines on routine non-IID fairness (where no single method dominates) and, most importantly, sharply reduces an adversarial client's aggregation weight under active fairness poisoning, lowering fair-client EOD where the data carries sufficient fairness signal (Section~\ref{sec:experiments}).
\section{Related Work}\label{sec:related}
Fairness interventions fall into three categories by where they act in the ML pipeline: in-processing (modify the training objective on-device), post-processing (post-hoc) (calibrate the trained model after aggregation), and hybrid (both). We position \fairis{} against each of them.

\subsection{In-processing, post-processing, and hybrid fairness schemes}
In-processing methods such as local fairness-constrained optimization and client-side reweighting~\cite{mehrabi2021survey,yang2024friends} reduce local bias during training but do not coordinate across clients: under non-IID data a locally fair model may combine into a globally biased aggregate, and pre-aggregation debiasing can be undone by the collaborative update (Section~\ref{sec:experiments}); they also offer no defense against a client that \emph{deliberately} maximizes disparity. Post-processing (post-hoc) methods such as threshold optimization~\cite{hardt2016equality} and demographic-parity calibration~\cite{mehrabi2021survey} adjust the decision boundary after aggregation, but require centralized access to predictions or a labeled validation set and address only prediction-time fairness, not a shared biased representation.

Hybrid schemes such as \fairfed{}~\cite{ezzeldin2023} adjust server-side aggregation weights by the gap between each client's local fairness score and the global average. This has two weaknesses: even with local debiasing, the aggregate can be re-biased by contributions from less fair clients; and, as we prove in Theorem~\ref{thm:ng}, the gap rule can be gamed under our threat model (Section~\ref{sec:threat}) by an adversary who matches the global score, keeping near-maximum weight while injecting bias. The latter is specific to that adversarial model; relative-gap weighting remains reasonable when all clients report honestly.

\fairis{} instead acts purely at the server-side aggregation step, weighting each client by its \emph{absolute} fairness score rather than its gap from a global average. This makes the rule ungameable in the above sense, guarantees positive weight for every client, and requires no change to client training; we describe it as a fairness-aware aggregation-weighting method that complements, rather than replaces, in-processing or post-processing debiasing. Like \fairfed{}, it assumes clients report their scores honestly; a forged low score understates bias, so the score channel requires \emph{verified} computation rather than authenticated transport, a threat we treat as out of scope (Section~\ref{sec:limitations}).

\subsection{Closely related methods}
The closest fairness-aware aggregators are \fairfed{}~\cite{ezzeldin2023} (gap between local and global EOD; Adult, COMPAS, $K{=}5$ clients) and GLocalFair~\cite{meerza2024glocalfair} (global and local fairness via clustering with a Gini surrogate; CelebA, Adult, UTK Faces; Appendix~\ref{app:background}). Of recent systems, SFFL~\cite{zhang2025sffl} targets equal-opportunity group fairness in heterogeneous FL (folktables U.S. Census tasks) by reweighting clients with a distance-based surrogate for the sensitive statistics it avoids collecting; it shares Fairis's no-disclosure goal but is full-model FL, reweights by update distance rather than a measured score, assumes honest participation, and offers no poisoning defense. SplitLPF~\cite{chen2024splitlpf} is a split-learning system but targets \emph{contribution} fairness, not group fairness, and its contribution weights can be inflated by an adversary reporting agreeable gradients. PrivEdge-SL~\cite{santhoshkumar2026privedgesl} is a single-client split-learning IoT framework with no cross-client aggregation or group-fairness objective, and is not a comparator. Among these, only SFFL targets group fairness without sensitive-attribute sharing, yet still offers no formal poisoning resistance: the gap \fairis{} fills.

FairTrade~\cite{badar2024fairtrade} instead treats fairness and balanced accuracy as competing objectives and searches for Pareto-optimal trade-offs. It optimizes where a client should sit on the fairness-accuracy frontier under honest participation, whereas we ask what an aggregation rule can guarantee when a participant is adversarial. We flag one point against our own interest: in the evaluation of Kasyap et al.~\cite{kasyap2025}, FairTrade contains the fairness-constrained optimization attack markedly better than \fairfed{} (disparity $0.049$ versus $0.442$ on Adult with one malicious client), so it, rather than \fairfed{}, is the strongest published baseline for this threat class. We do not re-implement it here because its multi-objective local optimization is not a drop-in aggregation rule and its published setting differs from ours, but a direct comparison against FairTrade is the most important experiment missing from this paper.

\subsection{Fairness poisoning attacks}
EAB-FL~\cite{meerza2024eabfl} is the first model poisoning attack specifically targeting group unfairness: it uses Layer-wise Relevance Propagation to identify redundant parameter spaces and injects bias through them while maintaining overall accuracy. Kasyap et al.~\cite{kasyap2025} independently propose a fairness-constrained optimization attack reaching a $90\,\%$ disparity increase at its strongest operating point, and report that \fairfed{} does not contain it. Neither attack has been evaluated against an absolute-score weighted scheme; to the best of our knowledge, \fairis{} is the first aggregation rule designed by construction to resist this attack class.

\subsection{Novelty summary}
Table~\ref{tab:novelty} (Appendix~\ref{app:novelty}) positions \fairis{} against six related methods on five dimensions: group fairness (GF), per-client scope (PC), architecture independence (AI), guaranteed positive weight (PW), and formal attack resistance (AR). \fairis{} is the only one combining all five. For AI, we mark $\checkmark$ any method whose weighting rule is \emph{validated} on two or more aggregation architectures; among those surveyed, only \fairis{} qualifies (validated on both FL and SplitML).
\section{Adversary Threat Model}\label{sec:threat}
We formalize the adversary's capabilities, define a game-based security notion, and state the three properties \fairis{} is designed to satisfy. Background definitions for collaborative learning architectures and group-fairness metrics are given in Appendix~\ref{app:background}.

\subsection{System and adversary model}
The system comprises $K\geq2$ clients and a central aggregation server running any collaborative learning protocol: Federated Learning (FL), Split Learning (SL), or SplitML. We index clients by $k\in\{1,\dots,K\}$ and communication rounds by $t\in\{1,\dots,T\}$. The case $K\geq2$ is assumed throughout, since the strict-monotonicity property below is vacuous for a single client (whose normalized weight is always $1$); with a minority adversary ($|\mathcal{B}|<K/2$, A1) the system in fact has $K\geq3$. Each client $C_k$ holds a private local dataset $\mathcal{D}_k$ and trains a local model $\mathcal{M}_k$ on that dataset. At each communication round, clients transmit a subset of model parameters $\theta_k$ (which may be the full model, a set of shared layers, or gradients, depending on the architecture) and a scalar local fairness score $\Fk \in [0,1]$ to the aggregation server. The server applies the aggregation protocol $\Pi$ to compute updated shared parameters and then broadcasts them to clients. The adversary $\mathcal{A}$ controls a minority subset $\mathcal{B}\subset\{C_k\}_{k=1}^K$ of malicious clients, with $|\mathcal{B}|<K/2$. $\mathcal{A}$ is \emph{probabilistic polynomial-time} ($\PPT$), meaning its strategy runs in time polynomial in a formal security parameter $\lambda$ instantiated by $(K,\eta)$.

\textbf{Capabilities.} $\mathcal{A}$ may observe the broadcast aggregated parameters $\theta$ after every round. In gap-based schemes such as \fairfed{}, $\mathcal{A}$ can also infer the server's aggregate fairness statistic $\Fg = \sum_k (n_k/\sum_j n_j)\,\Fk$, because it is a deterministic function of the reported local scores. Following \fairfed{}, we use this size-weighted mean of local scores as the server-side statistic. It is an \emph{estimator} of population-level EOD, not an identity: EOD is a ratio of group- and label-conditional counts and does not decompose as a client-size-weighted average of local EODs. Section~\ref{sec:limitations} records the consequence. $\mathcal{A}$ may set the training objective of each malicious client $C_k \in \mathcal{B}$ arbitrarily, subject to the utility constraint below.

\textbf{Adversarial goal.} The adversary aims to maximize group disparity injected into honest clients' shared representation while keeping prediction accuracy within $\varepsilon$ of the benign baseline, and maintaining the highest possible aggregation weight to amplify the attack's effect. At each round $t$, $\mathcal{A}$ treats the other clients' scores $\{\mathcal{F}_j^t\}_{j\notin\mathcal{B}}$ as fixed constants when choosing its own parameters; this is the algebraic condition the monotonicity argument requires (the other scores are held fixed while $\Fk$ varies) and does not assume the adversary can read those scores directly. Since local fairness scores are computed on private local data, they are not in general observable; where a scheme does expose a global summary (as in the gap-based \fairfed{} update below), we state that dependence explicitly. The optimization is over the malicious clients' parameters $\theta_\mathcal{B}$ only:
\begin{equation}\label{eq:goal}
\max_{\theta_\mathcal{B}}\;\sum_{k\in\mathcal{B}}\omega_k\cdot\Fk(\theta_k) \quad\text{s.t.}\quad \mathrm{Acc}(\theta_k)\geq\mathrm{Acc}_\mathrm{benign}-\varepsilon,\;\forall\,k\in\mathcal{B},
\end{equation}
where $\omega_k=(\eta-\Fk)/\sum_j(\eta-\mathcal{F}_j)$ under \fairis{}, with $\{\mathcal{F}_j\}_{j\notin\mathcal{B}}$ fixed at their observed values. The objective is nonlinear in the malicious scores, since each weight $\omega_k$ itself depends on $\Fk$. The accuracy constraint makes the attack stealthy: a model that maintains high prediction accuracy while producing unfair outputs cannot be detected by standard Byzantine defenses. This formulation captures both the EAB-FL attack~\cite{meerza2024eabfl} (which injects bias through redundant parameter space) and the Kasyap et al.\ attack~\cite{kasyap2025} (which uses fairness-constrained gradient manipulation).

\textbf{Assumptions.} The threat model rests on five explicit assumptions. (A1)~Malicious clients are a minority, $|\mathcal{B}|<K/2$. (A2)~Clients report their scalar fairness scores $\Fk$ honestly: the adversary controls each malicious client's model training but \emph{not} its fairness-score reporting. This assumption is critical, because a client that could arbitrarily falsify its score would circumvent the weighting mechanism entirely; such forged-score attacks are out of scope and require a verifiable score computation (Section~\ref{sec:limitations}). (A3)~The adversary is $\PPT$ and may observe every broadcast aggregate but cannot read other clients' private data or personalized parameters. (A4)~The server executes the aggregation protocol $\Pi$ faithfully (honest-but-curious server). (A5)~Each malicious client must keep its accuracy within $\varepsilon$ of the benign baseline, so the attack is stealthy with respect to accuracy-based defenses. All three security properties below are stated under these assumptions.

\subsection{Security game}
We define the security experiment $\mathsf{FairAtk}_{\mathcal{A},\Pi}(\lambda)$ for an aggregation protocol $\Pi$ and a $\PPT$ adversary $\mathcal{A}$. The challenger runs $K$ honest clients under protocol $\Pi$ for $T$ communication rounds and provides $\mathcal{A}$ with the aggregation rule. Here, $\Pi$ refers to any concrete aggregation scheme instantiated with its specific parameters: \fedavg{} (no fairness parameter), \fairfed{} with fairness budget $\beta$, or \fairis{} with inclusiveness parameter $\eta$. $\mathcal{A}$ selects one malicious client $C_0$ and sets its objective to solve Eq.~(\ref{eq:goal}). At round $T$, the experiment returns $(\omega_0^T,\mathcal{F}_0^T)$: the weight assigned to $C_0$ and the EOD of $C_0$'s model. Protocol $\Pi$ is \emph{attack-resistant} if $\omega_0^T$ is a strictly decreasing function of $\mathcal{F}_0^T$: a more biased attacker always receives less influence. This definition is intentionally framed around aggregation \emph{weight} rather than downstream honest-client EOD; the link from weight containment to fairness improvement is data-dependent and is examined empirically in Section~\ref{sec:experiments}.

\subsection{Security properties}
The following three properties together characterize a secure fairness-aware aggregation rule. A rule satisfying all three ensures that an adversary who increases group disparity automatically reduces its own influence, never excludes any client's demographic subgroup, and cannot recover influence by tracking the global fairness average.

\begin{property}[Monotone Weight Reduction (MWR)]\label{prop:mwr}
For all $\PPT$ adversaries $\mathcal{A}$ and all training rounds, $\partial\omega_k/\partial\Fk < 0$ when the other clients' scores $\{\mathcal{F}_j\}_{j\ne k}$ are held fixed (and $K\geq2$). A higher local fairness score (more bias) always yields strictly less aggregation weight.
\end{property}

\begin{property}[Demographic Participation (DP)]\label{prop:dp}
$\omega_k > 0$ for all clients $k$ and all rounds. No client is ever fully excluded from aggregation; each client's demographic subgroup remains represented in the shared model throughout training.
\end{property}

\begin{property}[Non-Gamesmanship (NG)]\label{prop:ng}
No $\PPT$ strategy simultaneously maintains $\Fk > 0$ and achieves the aggregation weight of a perfectly fair ($\Fk=0$) client. An adversary cannot recover the influence of a fair client without becoming fair itself.
\end{property}

\begin{remark}[Scope of the security properties]\label{rem:scope}
Properties~\ref{prop:mwr}--\ref{prop:ng} bound the adversary's aggregation \emph{weight} as a function of its bias; they make no direct claim about honest-client EOD, which additionally requires the honest majority to carry sufficient fairness signal (see Section~\ref{sec:experiments} for empirical evidence of when the two decouple).
\end{remark}
\section{Proposed Solution: \fairis{}}\label{sec:fairis}
Having established what a secure aggregation rule must satisfy, we introduce \fairis{} and show it satisfies all three security properties simultaneously. \fairis{} is architecture-agnostic: the aggregation server receives client parameters $\theta_k$ and scalar fairness scores $\Fk$, applies the Fairis weighting, and broadcasts the result, unchanged whether the protocol is FL or SplitML.

\subsection{Security parameter $\eta$}
The scalar $\eta>1$ is the practitioner-tunable security parameter. The only hard constraint is $\eta>1$, which guarantees $\wbar_k=\eta-\Fk>0$ for all $\Fk\in[0,1]$ (DP) and that the weight is strictly decreasing in $\Fk$ (MWR). We recommend the operating range $(1,(K{+}1)/K]$ as an empirical cutoff, not a requirement of the theory (which needs only $\eta>1$): by the decomposition in Eq.~(\ref{eq:limit}) the deviation $\omega_k-1/K$ scales as $O(1/\eta)$, so as $\eta$ grows the weights flatten toward the uniform share $1/K$ and the down-weighting of biased clients becomes marginal beyond this point. The most biased client's normalized weight is not in general near $1/K$ at $\eta=(K{+}1)/K$; by Eq.~(\ref{eq:limit}) a client more biased than the mean always sits strictly below $1/K$, by an amount that depends on the other clients' scores. We evaluate $\eta=1.01$ (high fairness emphasis) and $\eta=1.32$ (moderate); $\eta=1.32$ exceeds $(K{+}1)/K$ for $K\geq5$, so DP and MWR still hold but security differentiation weakens, and a practitioner should set $\eta$ per $K$. When $\eta$ is not specified by the practitioner, the following data-driven heuristic selects it from the observed fairness scores at the current round:
\begin{equation}\label{eq:eta}
s_1 = \frac{\min_j \mathcal{F}_j}{K},\qquad s_2 = \Bigl|\, s_1 - \tfrac{1}{K} \,\Bigr|,\qquad \eta = 1 + \max\!\bigl(\min(s_1,s_2),\,10^{-3}\bigr).
\end{equation}
The heuristic selects $\eta$ close to $1$ when the minimum observed score is small and larger under high heterogeneity; the $\max(\cdot,10^{-3})$ clamp keeps $\eta>1$ strictly even when $\min_j \mathcal{F}_j=0$, preserving DP. The heuristic is a practical default rather than an optimum; all reported experiments use the fixed values $\eta\in\{1.01,1.32\}$, and a sensitivity study of the heuristic is left to future work. Figure~\ref{fig:weight} shows how the weight on biased and fair clients varies with EOD and $\eta$.

\begin{figure}[ht]
  \centering
  \includegraphics[width=0.99\columnwidth]{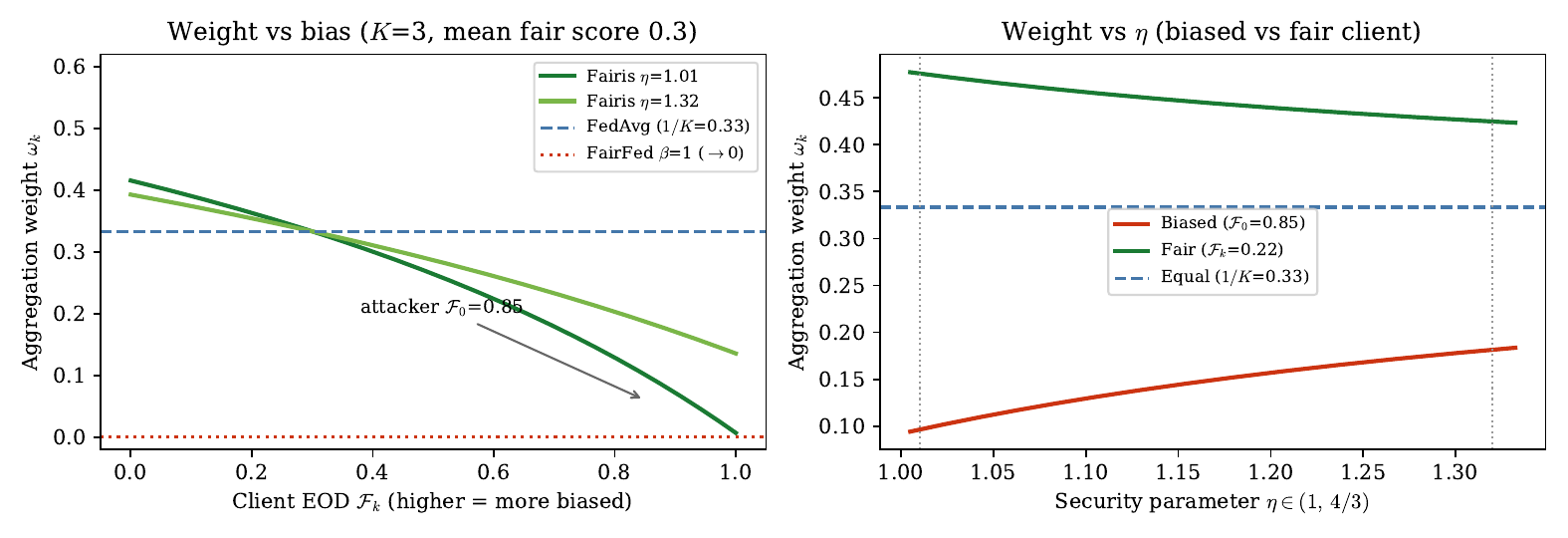}
  \caption{Effect of the security parameter $\eta$ on \fairis{} aggregation weights ($K=3$). \emph{Left}: aggregation weight $\omega_k$ as a function of a client's Equal Opportunity Difference (EOD) for two values of $\eta$, assuming two other clients with mean EOD$\,{=}\,0.30$. All \fairis{} curves are strictly decreasing (Theorem~\ref{thm:mwr}); smaller $\eta$ penalizes bias more sharply. \fedavg{} is constant at $1/K\,{=}\,0.33$ regardless of bias. \fairfed{} ($\beta=1$) drives highly biased clients towards zero weight (Demographic Participation violated for $\beta=1$; at $\beta=0$, \fairfed{} recovers \fedavg{} and DP holds trivially). \emph{Right}: weights plotted across the valid range $\eta\in(1,4/3)$ for a fixed biased client ($\mathcal{F}_0=0.85$) and a fair client ($\mathcal{F}_k=0.22$). The two experimental values $\eta\in\{1.01,1.32\}$ are marked with dotted lines; both are valid for $K=3$ since the upper bound is $(K+1)/K=4/3\approx1.33$.}
  \label{fig:weight}
\end{figure}

\paragraph{Large-$\eta$ limit (when \fairis{} reduces to \fedavg{}).} Let $\bar{\mathcal{F}}=\frac{1}{K}\sum_{j}\mathcal{F}_j$ denote the mean local score, and note that the normalizing denominator satisfies $\sum_{j}(\eta-\mathcal{F}_j)=K\eta-\sum_{j}\mathcal{F}_j$. The normalized weight then decomposes exactly as
\begin{equation}\label{eq:limit}
\omega_k \;=\; \frac{\eta-\mathcal{F}_k}{\sum_{j}(\eta-\mathcal{F}_j)}
\;=\; \frac{1}{K} \;+\; \frac{\bar{\mathcal{F}}-\mathcal{F}_k}{\,K\eta-\sum_{j}\mathcal{F}_j\,},
\end{equation}
a uniform $1/K$ term plus a correction whose numerator is the signed gap $\bar{\mathcal{F}}-\mathcal{F}_k$ (positive for fairer-than-average clients, negative for more-biased ones) and whose denominator is $K\eta-\sum_{j}\mathcal{F}_j>0$. The correction is $O(1/\eta)$; hence $\omega_k\to1/K$ as $\eta\to\infty$ and \fairis{} converges to the \emph{uniform} average. This limit is not \fedavg{}: \fedavg{} weights are size-proportional, $n_k/\sum_j n_j$, and the two coincide only in the special case where every client holds equally many records. The collapse is also one-sided: no \emph{finite} $\eta$ reproduces \fairfed{}'s gap-based rule for $\beta>0$, and since \fairfed{} at $\beta=0$ is size-proportional \fedavg{} while \fairis{} at $\eta\to\infty$ is uniform, the two families share a common limit only under equal client sizes. Because the deviation from the uniform share is the signed gap $\bar{\mathcal{F}}-\Fk$ divided by $K\eta-\sum_j\mathcal{F}_j$, increasing $\eta$ smoothly erases the fairness differentiation that the recommended bound $(K{+}1)/K$ preserves.

\paragraph{Per-$K$ guidance.} The upper bound $(K{+}1)/K$ is $4/3$ for $K=3$, $1.2$ for $K=5$, and $1.1$ for $K=10$. Thus $\eta=1.01$ is in range for all $K$, whereas $\eta=1.32$ is in range only for $K=3$; for $K\geq5$ it exceeds the bound (DP and MWR still hold, but security differentiation weakens), and such cases are flagged as out-of-range in the tables.

\subsection{Size-weighted variant (\fairisn{})}\label{sec:fairisn}
The rule of Eq.~(\ref{eq:limit}) weights purely by fairness and discards the client-size information that \fedavg{} uses. That is a deliberate simplification, but it has two costs that our ablation makes visible (Section~\ref{sec:standard}): a small client and a large client with equal scores receive equal influence, and when all scores coincide the rule degenerates to uniform averaging rather than to \fedavg{}. We therefore also evaluate the size-weighted variant
\begin{equation}\label{eq:fairisn}
\omega_k^{\textsc{n}} \;=\; \frac{n_k\,(\eta-\Fk)}{\sum_j n_j\,(\eta-\mathcal{F}_j)},
\end{equation}
which multiplies the unnormalized \fairis{} score by the client's sample count.

\begin{proposition}[Properties carry over to \fairisn{}]\label{prop:fairisn}
For $\eta>1$, $K\geq2$, and $n_k>0$ for all $k$, the rule of Eq.~(\ref{eq:fairisn}) satisfies Monotone Weight Reduction (Theorem~\ref{thm:mwr}), Demographic Participation (Theorem~\ref{thm:dp}), Non-Gamesmanship (Theorem~\ref{thm:ng}), and the influence bound of Proposition~\ref{prop:influence} with $\omega_0$ replaced by $\omega_0^{\textsc{n}}$.
\end{proposition}

\begin{proof}
Each $n_k$ is a positive constant with respect to every $\mathcal{F}_j$, so writing $W=\sum_j n_j(\eta-\mathcal{F}_j)>0$ we obtain $\partial\omega_k^{\textsc{n}}/\partial\Fk = -n_k\bigl(W-n_k(\eta-\Fk)\bigr)/W^2 = -n_k\sum_{j\ne k}n_j(\eta-\mathcal{F}_j)/W^2<0$ for $K\geq2$, giving MWR. Positivity is immediate from $n_k>0$ and $\eta-\Fk>0$, giving DP. NG holds because $\omega_k^{\textsc{n}}$ depends on the client's own \emph{absolute} score and never on its distance from a population statistic, so the fixed point exploited against \fairfed{} does not exist; strict monotonicity then forbids any biased client from attaining the weight of a perfectly fair one. The influence bound follows from clipping and DP exactly as in Proposition~\ref{prop:influence}.
\end{proof}

Two structural consequences distinguish \fairisn{} from \fairis{}. First, as $\eta\to\infty$ the weights converge to $n_k/\sum_j n_j$, so the large-$\eta$ limit is \fedavg{} proper rather than uniform averaging. Second, when all clients report the same score the rule reduces to \fedavg{} rather than to $1/K$, so the saturated regime degrades to the standard baseline instead of to a size-blind one. \fairisn{} therefore removes the confound between fairness weighting and the removal of size proportionality while preserving every guarantee; whether it also improves measured fairness is an empirical question addressed in Section~\ref{sec:standard}.

\subsection{Aggregation algorithm}
Algorithm~\ref{alg:fairis} states the complete \fairis{} protocol in general form, applicable to any collaborative learning architecture. Clients train locally on their private data and compute a local fairness score. They transmit their shared parameters $\theta_k$ (the full model in FL, the shared layers in SplitML, or any other designated subset) and the scalar score $\Fk$ to the server. The server computes the Fairis weights and broadcasts the weighted aggregate. Four key properties follow directly from the weight formula: (i)~positive weights, since $\eta>1$ and $\Fk\in[0,1]$ imply $\wbar_k = \eta-\Fk>0$ always (DP satisfied); (ii)~monotone weighting, since $\partial\omega_k/\partial\Fk < 0$ (MWR satisfied); (iii)~no client pre-processing required, as clients need only compute their fairness score after standard local training; and (iv)~architecture independence, since the weighting formula operates only on the fairness scores and is indifferent to the structure of $\theta_k$.

\fairis{} adds no asymptotic overhead over \fedavg{}: the server performs $O(K)$ extra scalar operations per round ($K$ subtractions, one sum, $K$ divisions), independent of model size, and each client reuses the forward pass already performed during evaluation to compute its score. Appendix~\ref{app:complexity} gives a full account, including why \fedavg{}-style convergence is expected while a formal rate under round-varying fairness weights remains an open question.

\begin{algorithm}[ht]
\caption{\fairis{}: fairness-score-weighted aggregation for any collaborative learning setting}
\label{alg:fairis}
\begin{algorithmic}[1]
\REQUIRE Clients $\{C_k\}_{k=1}^K$ ($K\geq2$), fairness metric $\mathcal{F}$, security parameter $\eta>1$ (recommended range $(1,(K+1)/K]$), clipping bound $C>0$
\STATE $\theta\leftarrow$ common initialization
\FOR{each training round $t$}
  \FOR{each client $C_k$}
    \STATE Train local model $\mathcal{M}_k$ on $\mathcal{D}_k$ starting from $\theta$
    \STATE Compute local fairness score $\Fk = \mathcal{F}(\mathcal{M}_k, \mathcal{D}_k) \in [0,1]$
    \STATE Transmit update $\Delta_k\leftarrow\theta_k-\theta$ and $\Fk$ to server
  \ENDFOR
  \IF{$\eta$ not pre-specified}
    \STATE $\eta\leftarrow\max(\min(s_1,s_2),10^{-3})+1$ via Eq.~(\ref{eq:eta})
  \ENDIF
  \FOR{each client $C_k$}
    \STATE $\Delta_k\leftarrow\Delta_k\cdot\min\!\left(1,\;C/\lVert\Delta_k\rVert_2\right)$ \hfill\COMMENT{norm clipping: $\lVert\Delta_k\rVert_2\le C$}
    \STATE $\wbar_k\leftarrow\eta-\Fk$ \hfill\COMMENT{guaranteed $> 0$ since $\eta>1$, $\Fk\le1$}
  \ENDFOR
  \STATE $\omega_k\leftarrow\wbar_k/\sum_j\wbar_j$ for all $k$
  \STATE $\theta\leftarrow\theta+\sum_k\omega_k\cdot\Delta_k$;\quad broadcast $\theta$ to all clients
\ENDFOR
\end{algorithmic}
\end{algorithm}

A worked example of the weighting and a detailed account of five structural limitations of \fairfed{} that \fairis{} resolves are given in Appendix~\ref{app:worked} and Appendix~\ref{app:fairfed_limits}.
\section{Security Analysis}\label{sec:security}
All three security properties of \fairis{} follow from the closed-form weight formula $\omega_k = (\eta-\Fk)/\sum_j(\eta-\mathcal{F}_j)$; no cryptographic assumption beyond standard $\PPT$ complexity is required. The proofs are elementary algebra over the weight map; we call the resulting facts \emph{security} properties because they are precisely the primitives an aggregation rule needs to contain a fairness-poisoning adversary (Section~\ref{sec:threat}), not because they rest on a cryptographic hardness assumption. For space, they are deferred to Appendix~\ref{app:proofs}.

\paragraph{Machine-checked proofs.} The results of this section have been
independently formalized and checked in the Lean proof assistant by one of the
authors. The development covers Theorems~\ref{thm:mwr}, \ref{thm:dp} and
\ref{thm:ng} and Corollary~\ref{cor:coalition} for \fairis{}; the corresponding
participation, normalization, monotonicity and non-gamesmanship results for
\fairisn{} (Proposition~\ref{prop:fairisn}); $\eta>1$ as the \emph{exact}
condition for strictly positive unnormalized mass at every score in $[0,1]$;
the clipped displacement bound of Proposition~\ref{prop:influence}, together
with the accompanying negative result that no finite uniform bound holds
without clipping; a concrete instance in which every local EOD and their
client-size-weighted average are zero while the pooled EOD is $49/50$, which is
the counterexample underlying limitation~(iii) of
Section~\ref{sec:limitations}; and the \fairfed{} matched-gap fixed point at
the honest clients' size-weighted mean, with the exact geometric error formula
for repeated matching. Formalization establishes that the algebra is correct;
it does not validate the modeling. Assumptions (A1) and (A2) are premises of
the Lean statements rather than consequences of them, and the sharpest
limitation of this work, honest score reporting, is untouched by it.

\begin{theorem}[Monotone Weight Reduction (MWR)]\label{thm:mwr}
Under \fairis{} with $K\geq2$ clients, the aggregation weight $\omega_k = (\eta-\Fk)/\sum_j(\eta-\mathcal{F}_j)$ is strictly decreasing in $\Fk$ for every $k$, holding the other clients' scores $\{\mathcal{F}_j\}_{j\ne k}$ fixed. (For $K=1$ the normalized weight is constant at $\omega_k\equiv1$, so $K\geq2$ is necessary for strictness.)
\end{theorem}

\begin{corollary}[Worst-case weight bound]\label{cor:opt}
Under \fairis{}, a $\PPT$ adversary maximizes its aggregation weight by setting $\Fk=0$, i.e.\ by becoming perfectly fair. Any positive group-disparity strictly and continuously reduces the adversary's weight below this maximum. Note that setting $\Fk=0$ achieves the maximum \emph{weight} but zero \emph{attack impact} (Eq.~\ref{eq:goal}); the optimal attack strategy balances these two objectives.
\end{corollary}

\begin{corollary}[Coalition weight bound]\label{cor:coalition}
Under \fairis{}, the total aggregation weight of any colluding minority coalition $\mathcal{B}$ ($|\mathcal{B}|<K/2$) is strictly decreasing in each member's bias. Formally, letting $W_\mathcal{B}=\sum_{k\in\mathcal{B}}(\eta-\mathcal{F}_k)$ and $W_{\neg\mathcal{B}}=\sum_{k\notin\mathcal{B}}(\eta-\mathcal{F}_k)$, we have for every $k'\in\mathcal{B}$:
\[
\frac{\partial}{\partial\mathcal{F}_{k'}}\!\sum_{k\in\mathcal{B}}\omega_k
= \frac{-W_{\neg\mathcal{B}}}{(W_\mathcal{B}+W_{\neg\mathcal{B}})^2} < 0,
\]
since $W_{\neg\mathcal{B}}>0$ (honest majority guaranteed by~(A1)). The coalition's maximum total weight is achieved only when all members become perfectly fair, in which case the attack has zero bias impact (Eq.~\ref{eq:goal}). MWR thus extends to colluding coalitions without additional assumptions, formally bounding the adversarial threat even under coordinated multi-client attacks.
\end{corollary}

\begin{theorem}[Demographic Participation (DP)]\label{thm:dp}
Under \fairis{} with $\eta>1$, $\omega_k>0$ for all clients $k$ and all fairness scores $\Fk\in[0,1]$.
\end{theorem}

MWR constrains an adversary's aggregation \emph{weight}, but weight alone does not constrain its effect on the global model: in $\theta\leftarrow\theta+\sum_k\omega_k\Delta_k$ an adversary can offset a small $\omega_0$ by inflating $\lVert\Delta_0\rVert$. The norm clipping in Algorithm~\ref{alg:fairis} closes exactly this gap, and turns MWR into a bound on displacement of the global model.

\begin{proposition}[Influence containment under clipping]\label{prop:influence}
Let the server clip every client update to $\lVert\Delta_k\rVert_2\le C$ before aggregation, as in Algorithm~\ref{alg:fairis}. Then the adversary $C_0$'s contribution to the global update satisfies
\[
\bigl\lVert\omega_0\Delta_0\bigr\rVert_2 \;\le\; \omega_0\,C \;=\; \frac{C\,(\eta-\mathcal{F}_0)}{\sum_j(\eta-\mathcal{F}_j)},
\]
which by Theorem~\ref{thm:mwr} is strictly decreasing in $\mathcal{F}_0$. Increasing injected group disparity therefore strictly reduces the adversary's maximum achievable displacement of the global model, uniformly over its choice of $\Delta_0$.
\end{proposition}

\begin{proof}
Clipping gives $\lVert\Delta_0\rVert_2\le C$, so $\lVert\omega_0\Delta_0\rVert_2=\omega_0\lVert\Delta_0\rVert_2\le\omega_0C$ since $\omega_0>0$ by Theorem~\ref{thm:dp}. Substituting the closed form $\omega_0=(\eta-\mathcal{F}_0)/\sum_j(\eta-\mathcal{F}_j)$ gives the stated expression, and strict monotonicity in $\mathcal{F}_0$ is Theorem~\ref{thm:mwr}; multiplying by the constant $C>0$ preserves it.
\end{proof}

Without clipping no such bound exists, since $\lVert\omega_0\Delta_0\rVert_2$ is unbounded above for any fixed $\omega_0>0$. Clipping is therefore not an optional implementation detail but a precondition for the containment claim. Empirically it is inexpensive: calibrated per configuration to the $90$th percentile of honest update norms measured in an adversary-free pilot, the bound $C$ lies between $2.83$ and $5.41$ and the clip binds on $9.0$ to $30.6\,\%$ of honest updates across aggregation rules (median $12.5\,\%$), so honest training is only mildly constrained (Section~\ref{sec:experiments}).

\begin{theorem}[Non-Gamesmanship (NG)]\label{thm:ng}
Under \fairis{}, no $\PPT$ strategy simultaneously maintains $\Fk>0$ and achieves the aggregation weight of a perfectly fair ($\Fk=0$) client. Under \fairfed{}, a $\PPT$ adversary who observes $\Fg$ and matches it ($\Fk=\Fg>0$) holds a zero fairness gap while maintaining positive group disparity; under the unclipped gap update with budget $\beta>0$, positive total weight, and a positive average gap $\bar\Delta>0$, the adversary's normalized weight does not decrease but strictly increases. The zero-flooring used in practice does not rescue the rule: the floor can bind only on clients whose weight is being driven down, never on the matched-gap adversary whose weight is increasing (Appendix~\ref{app:proofs}).
\end{theorem}

\paragraph{Scope of the formal guarantees.} Theorems~\ref{thm:mwr}--\ref{thm:ng} and Corollary~\ref{cor:coalition} bound adversarial aggregation \emph{weight}; they make no claim about the absolute EOD of honest clients (which additionally requires honest clients to carry sufficient fairness signal, as shown in Section~\ref{sec:experiments}). Assumption (A2) is a trust assumption on the score channel; authentication or MPC is required in fully adversarial deployments (Section~\ref{sec:limitations}). The name \emph{Non-Gamesmanship} refers specifically to the impossibility of recovering a perfectly fair client's weight while remaining biased (Theorem~\ref{thm:ng}); it is not a claim that no manipulation of any kind is possible. The \fairfed{} clause exhibits one concrete gaming strategy, matching the global score, that \fairis{} removes.

\textbf{Empirical security evaluation.} Under active fairness poisoning at $\alpha_\text{adv}=2$, \fedavg{} gives the attacker its full data-proportional share regardless of injected bias ($0.395$, $0.210$, $0.303$ on German, Taiwan, Adult), and Uniform fixes it at $1/K=0.333$ on all three, confirming that neither rule reacts to the fairness signal. \fairfed{} ($\beta=1$) drives it to exactly $0$ on Taiwan ($0.380$, $0.000$, $0.180$), removing that client's subgroup and violating DP; its gap rule is not lower-bounded, so positivity is never guaranteed. \fairis{} ($\eta=1.01$) holds the attacker to $0.005$, $0.048$, $0.183$, the lowest weight of any rule on German and Taiwan and within $0.003$ of clipped \fairfed{} on Adult ($0.180$), while alone guaranteeing strict positivity (Theorem~\ref{thm:dp}) and a monotone influence bound (Proposition~\ref{prop:influence}); at $\eta=1.32$ it is $0.108$, $0.154$, $0.230$. The comparison at a \emph{stealthy} attack strength, which is the regime the threat model actually specifies, is deferred to Section~\ref{sec:sweep}. Figure~\ref{fig:mwr} (Appendix) plots these values and Table~\ref{tab:attack} lists them.
\section{Experiments}\label{sec:experiments}
We evaluate \fairis{} on three financial datasets under both routine non-IID heterogeneity and active fairness poisoning, comparing it against Silo, \fedavg{}, and \fairfed{}. We first describe the setup, then report the standard non-IID and active-poisoning results.

\subsection{Setup}
This subsection describes the datasets, the partitioning and evaluation protocol, the compared methods, and the poisoning configuration.

\textbf{Datasets.} We evaluate on three datasets of differing size and fairness profile. \emph{German Credit} (Statlog)~\cite{hofmann1994}: 1,000 records, 20 features; we group on the personal-status-and-sex attribute, taking single or divorced/separated males (codes A91 and A93~\cite{hofmann1994}) as unprivileged and all others as privileged (common alternatives group by sex or by age); its small size stresses the cross-silo setting. \emph{Taiwan Credit} (Default of Credit Card Clients)~\cite{yeh2016}: 30,000 records, 23 features, sensitive attribute sex, with a pronounced 22\,\% default-rate imbalance. \emph{Adult Income}~\cite{becker1996adult}: 30,162 complete records, 13 features, sensitive attribute sex; the primary benchmark of \fairfed{}~\cite{ezzeldin2023} and EAB-FL~\cite{meerza2024eabfl}.

\textbf{Partitioning and evaluation protocol.} We partition each dataset across $K$ clients with $\mathrm{Dir}(\alpha=0.5)$ over the four (label, sensitive attribute) cells, so both label balance and group composition vary across clients. A stratified 20\,\% global test pool is held out before partitioning and used to evaluate each client's model, ensuring fairness metrics remain well-defined for small local sets. Metrics are mean $\pm$ standard deviation (std) over non-degenerate client-run combinations across three seeds; degenerate (single-class) predictors are excluded, at a rate that is method-independent.

\textbf{Methods.} We compare four methods (six configurations): \textbf{Silo} (no collaboration; utility lower bound), \fedavg{}~\cite{mcmahan2017} (dataset-size-weighted, no fairness mechanism), \fairfed{}~\cite{ezzeldin2023} with $\beta\in\{0,1\}$ ($\beta=0$ recovers \fedavg{}, $\beta=1$ is the maximum fairness budget), and \fairis{} with $\eta\in\{1.01,1.32\}$ (high and moderate fairness emphasis; the attacker keeps $0.5\,\%$ to $16.5\,\%$ and $10.8\,\%$ to $21.7\,\%$ of the weight respectively). The value $\eta=1.01$ lies in the recommended range $(1,(K{+}1)/K]$ for all $K$; $\eta=1.32$ is in range only for $K=3$ and is retained across $K$ for a consistent comparison (Section~\ref{sec:fairis}). GLocalFair~\cite{meerza2024glocalfair} and SFFL~\cite{zhang2025sffl} are discussed but not re-implemented (Section~\ref{sec:limitations}). All clients train identical 4-layer MLPs (input$\to$64$\to$32$\to$16$\to$1, sigmoid) with Adam (learning rate $5\times10^{-3}$, up to 60 steps per round, 16 rounds) and class-balanced cross-entropy; SplitML on the shared top 2 layers (input side), FL all 4 layers. Seeds are 42, 123, 456; we report mean$\pm$std in tables and standard error in figures, treating sub-variability differences as indicative.

\textbf{Active fairness poisoning.} To simulate the EAB-FL and Kasyap threat model, client $C_0$ trains with the adversarial objective $\mathcal{L}_\text{total} = \mathcal{L}_\text{CE} - \alpha_\text{adv}|\bar{p}_{A=1}-\bar{p}_{A=0}|$ with $\alpha_\text{adv}=2$, maximizing the difference in predicted positive-class probabilities between demographic groups while maintaining overall prediction accuracy. The remaining clients $C_1,\ldots,C_{K-1}$ train normally on balanced data.

\subsection{Standard non-IID results}\label{sec:standard}
Table~\ref{tab:main} reports EOD and Statistical Parity Difference (SPD) at $K=10$ for each dataset, with the full per-$K$ EOD breakdown in Figure~\ref{fig:eod} (Appendix~\ref{app:eod_fig}). Figure~\ref{fig:fl_vs_splitml} (Appendix) reports the key architecture-agnostic comparison: \fairis{} under full-model FL aggregation (hatched bars) versus SplitML partial-model aggregation (solid bars).

\textbf{SplitML results.} On \textbf{German Credit} at $K=10$, \fairis{} ($\eta=1.01$) attains the lowest observed mean EOD ($0.112\pm0.109$) and SPD ($0.129$), below \fedavg{} ($0.144$). The Uniform ablation shows this is \emph{not} evidence of fairness differentiation: Uniform attains bitwise identical values, because on this small dataset every client interpolates its local split, the score channel saturates to zero, \fairfed{} collapses exactly onto \fedavg{}, and \fairis{} degenerates to uniform averaging. The entire German advantage over \fedavg{} is therefore attributable to discarding size-proportional weighting, not to down-weighting biased clients. On \textbf{Taiwan Credit} at $K=10$, \fairis{} does not lead: \fairfed{} ($\beta{=}1$) attains the lowest mean EOD ($0.382$), Silo follows ($0.401$), and \fairis{} ($\eta=1.01$) is higher ($0.467$), above both \fedavg{} ($0.426$) and Uniform ($0.445$). Silo's lead here illustrates that collaboration itself can import bias from other clients' data; the \fairis{} guarantee is containment of adversarial influence (Remark~\ref{rem:scope}), not a universal fairness improvement over local training. On \textbf{Adult Income} at $K=10$, \fairfed{} ($\beta=1$) attains the lowest mean EOD ($0.173$) and SPD ($0.167$), with \fairis{} above it ($0.251$ and $0.254$) and also above Uniform ($0.258$ EOD) only marginally. Standard deviations are wide across all three datasets (imbalanced, small per-client splits), and no single method dominates the routine non-IID setting; the method's distinctive behavior instead appears under active poisoning (Section~\ref{sec:security}), where it alone contains the adversary's aggregation weight.

\textbf{FL full-model results (architecture-agnostic validation).} We apply the same aggregation rule unchanged to homogeneous full-model FL, where all parameters are aggregated (Table~\ref{tab:fl}). Here the picture differs from SplitML: on German and Taiwan Credit the local score channel largely collapses under full-model averaging, so \fedavg{} (equivalently \fairfed{} $\beta{=}0$) attains the lowest observed mean FL EOD at $K=5$ ($0.098$ on German, $0.161$ on Taiwan) and \fairis{} sits above it ($0.252$ and $0.231$ for $\eta=1.01$), matching Uniform exactly on German for the saturation reason above; on \textbf{Adult Income}, \fairfed{} $\beta{=}1$ attains the lowest mean FL EOD ($0.150$) and \fairis{} ($\eta=1.01$) reaches $0.184$, both improving on \fedavg{} ($0.209$) at comparable accuracy ($0.783$ vs.\ $0.804$). \fairis{} thus applies without modification across both regimes and never collapses the aggregation, keeping every client's weight strictly positive. The point of this experiment is architecture-agnostic \emph{transfer}: the identical closed-form weighting runs on both SplitML and full-model FL without any change (Figure~\ref{fig:fl_vs_splitml}).

A round-by-round comparison of global-model fairness across communication rounds is given in Appendix~\ref{app:convergence}.

\subsection{Active fairness poisoning}
Table~\ref{tab:attack} reports the security evaluation under active fairness poisoning ($K=3$, client $C_0$ trains adversarially, clients $C_1$ and $C_2$ train normally, 2 shared layers, mean over 3 seeds). The table includes all \fairfed{} and \fairis{} configurations to show the security-inclusiveness trade-off.

The adversarial client reaches an EOD of roughly $0.39$ to $0.86$ across aggregation rules. Accuracy above chance is not by itself evidence of stealth, so we measure the utility constraint of Section~\ref{sec:threat} directly in Section~\ref{sec:sweep} rather than assuming it. This probes the high-bias end of the curve; the full monotone relationship across $\mathcal{F}_k$ is established in Theorem~\ref{thm:mwr} and Figure~\ref{fig:weight}.

The attacker-weight reductions (Section~\ref{sec:security}, Table~\ref{tab:attack}) translate into honest-client protection only where the data carries a fairness signal. On German Credit, \fairis{} ($\eta=1.01$) lowers fair-client EOD to $0.364$, below \fedavg{} ($0.390$); \fairfed{} ($\beta=1$) reaches a marginally lower $0.333$ but only by heavily down-weighting the attacker toward removal. On Taiwan Credit, structural class imbalance keeps every method's fair-client EOD in the $0.41$ to $0.54$ range (Remark~\ref{rem:scope}), so the benefit appears as the $77\,\%$ attacker-weight reduction ($0.210\to0.048$), not in EOD. On Adult Income, fair-client EOD stays elevated for every method and \fairis{} confers no honest-client advantage, which Section~\ref{sec:sweep} explains: on that dataset the honest clients are themselves highly unfair, so the adversary is not a score outlier and the mechanism has nothing to act on.

\subsection{Attack strength, stealth, and the scope of containment}\label{sec:sweep}
The threat model of Section~\ref{sec:threat} requires the adversary to stay within $\epsilon$ accuracy of a benign baseline; this is what makes fairness poisoning invisible to accuracy-based Byzantine defenses, and it is the premise the attack inherits from EAB-FL~\cite{meerza2024eabfl} and Kasyap et al.~\cite{kasyap2025}. Rather than assume it, we measure it. We sweep the attack strength $\alpha_\text{adv}\in\{0.25,0.5,1,2\}$ and record, at every setting, both the adversary's aggregation weight and its accuracy gap against an adversary-free control run on the same partition and seed (Table~\ref{tab:sweep}).

Two points of methodology. First, the correct control for \fairis{} is \emph{uniform} weighting, not \fedavg{}: \fairis{} and uniform weighting both ignore client size, whereas \fedavg{} does not, so measuring \fairis{} against \fedavg{} conflates fairness weighting with the removal of size proportionality. Second, we report the accuracy gap rather than raw accuracy, because only the gap speaks to detectability.

\textbf{A stealthy adversary is contained on Taiwan Credit.} At $\alpha_\text{adv}=0.5$ the adversary loses $0.006$ accuracy relative to control, which no accuracy-based defense could plausibly detect, while driving its model EOD to $0.550$. \fairis{} ($\eta=1.01$) nonetheless cuts its aggregation weight to $0.198$, a $41\,\%$ reduction below the size-blind control ($1/K=0.333$); at $\alpha_\text{adv}=1$ the reduction is $54\,\%$ at a $0.038$ gap. This is the paper's central empirical claim, and on this dataset it holds in the regime that matters: the adversary is invisible to accuracy-based defenses and is still caught by the fairness channel.

\textbf{Containment fails when the honest population is already unfair.} On Adult Income at $\alpha_\text{adv}\le1$, \fairis{} assigns the adversary \emph{more} weight than uniform ($0.343$, $0.335$, $0.336$ against $0.333$). This is not noise. By the decomposition of Eq.~(\ref{eq:limit}), $\omega_k = 1/K + (\bar{\mathcal{F}}-\mathcal{F}_k)/(K\eta-\sum_j\mathcal{F}_j)$ with $K\eta-\sum_j\mathcal{F}_j>0$, so $\omega_0>1/K$ holds \emph{if and only if} $\mathcal{F}_0<\bar{\mathcal{F}}$. On Adult the adversary's reported local EOD sits below the honest clients' mean: the honest clients are themselves highly unfair locally, a mildly adversarial client does not stand out, and absolute-score weighting rewards it for looking comparatively fair. Only at $\alpha_\text{adv}=2$, where the adversary becomes a genuine outlier, does containment appear, and by then its accuracy gap is $0.201$ and it is detectable by other means.

\textbf{Scope condition.} Containment strength tracks how far the adversary's local score separates from the honest population's mean, not how adversarial it is in absolute terms. \fairis{} contains an adversary that is an outlier in the score distribution and does nothing when the honest population is uniformly unfair. German Credit at $\alpha_\text{adv}=0.25$ illustrates a third case: a weak attack barely raises disparity at all (EOD $0.306$, close to honest), so there is little to contain and the reduction is $2.8\,\%$.

\textbf{The design tension this exposes.} Absolute-score weighting is ungameable (Theorem~\ref{thm:ng}) precisely because it ignores the population distribution, and that same property is why it fails on Adult. Relative weighting, as in \fairfed{}'s gap rule, adapts to the population but is gameable for exactly the same reason. We state this trade-off explicitly rather than claim unconditional containment: a rule that is robust to a uniformly unfair population without reintroducing gameability is, in our view, the natural next problem.

\begin{table}[ht]
\centering
\caption{Attack-strength sweep ($K=3$, SplitML, mean over 3 seeds). Accuracy gap is $|\text{acc}_\text{atk}-\text{acc}_\text{benign}|$ against an adversary-free control on the same partition and seed; rows with a gap at or below $0.05$ are stealthy in the sense of Section~\ref{sec:threat} and are marked $\star$. Reduction is \fairis{} ($\eta{=}1.01$) against the size-blind control (uniform, $1/K=0.333$). Negative reduction means the adversary received \emph{more} weight than uniform, which by Eq.~(\ref{eq:limit}) happens exactly when its reported score falls below the honest mean.}
\label{tab:sweep}
\begin{tabular}{llccccc}
\toprule
Dataset & $\alpha_\text{adv}$ & Uniform $\omega_0$ & \fairis{} $\omega_0$ & Reduction & Acc.\ gap & Atk.\ EOD \\
\midrule
\multirow{4}{*}{German} & 0.25 & 0.333 & 0.324 & 2.8\,\% & 0.040$^\star$ & 0.306 \\
 & 0.5 & 0.333 & 0.224 & 32.8\,\% & 0.063 & 0.314 \\
 & 1 & 0.333 & 0.118 & 64.5\,\% & 0.070 & 0.428 \\
 & 2 & 0.333 & 0.005 & 98.5\,\% & 0.145 & 0.547 \\
\midrule
\multirow{4}{*}{Taiwan} & 0.25 & 0.333 & 0.268 & 19.6\,\% & 0.012$^\star$ & 0.527 \\
 & 0.5 & 0.333 & 0.198 & \textbf{40.6\,\%} & 0.005$^\star$ & 0.550 \\
 & 1 & 0.333 & 0.152 & \textbf{54.4\,\%} & 0.038$^\star$ & 0.567 \\
 & 2 & 0.333 & 0.048 & 85.5\,\% & 0.111 & 0.692 \\
\midrule
\multirow{4}{*}{Adult} & 0.25 & 0.333 & 0.343 & $-$2.8\,\% & 0.057 & 0.531 \\
 & 0.5 & 0.333 & 0.335 & $-$0.6\,\% & 0.183 & 0.626 \\
 & 1 & 0.333 & 0.336 & $-$0.8\,\% & 0.238 & 0.440 \\
 & 2 & 0.333 & 0.183 & 45.0\,\% & 0.201 & 0.387 \\
\bottomrule
\end{tabular}
\end{table}

\subsection{The gamesmanship attack}\label{sec:gaming}
The preceding experiments use a disparity-maximizing adversary, which is the regime in which \fairfed{} performs comparatively well. Theorem~\ref{thm:ng} predicts a different and more damaging attack against gap-based weighting, and this section tests it directly.

\textbf{Attacker model.} The adversary does not maximize disparity. It \emph{steers} its local disparity toward the aggregate fairness statistic $\Fg$ broadcast in the previous round, minimizing $\mathcal{L}_\text{CE}+\lambda\bigl|\,\lvert\bar p_{A=1}-\bar p_{A=0}\rvert-\Fg\,\bigr|$ with $\lambda=4$ and a round-0 target of $0.3$ before any global statistic exists. Crucially the adversary still reports its \emph{honestly computed} local EOD, so assumption~(A2) is respected throughout: it manipulates its training, never its reporting. Under \fairfed{} this drives its fairness gap $\Delta_0=|\Fg-\mathcal{F}_0|$ toward zero while its disparity remains positive, which is exactly the configuration for which the gap update $\wbar_k\leftarrow\wbar_k-\beta(\Delta_k-\bar\Delta)$ \emph{increases} the adversary's weight. Under \fairis{} and \fairisn{} the same behavior simply registers as a positive absolute score and is penalized by MWR, because neither rule references $\Fg$.

\textbf{Protocol.} $K=3$ clients, SplitML mode with 2 shared layers, clipping enabled, three seeds, $16$ rounds, with client $C_0$ adversarial.

\textbf{The attack has a fixed point, and it bounds its own strength.} Because the aggregate statistic $\Fg=\sum_k(n_k/n)\mathcal{F}_k$ includes the adversary's own score, the matched-gap condition $\mathcal{F}_0=\Fg$ is self-referential. Solving it,
\begin{equation}\label{eq:gamefix}
\mathcal{F}_0^{\ast}\;=\;\frac{\sum_{k\ne0}n_k\mathcal{F}_k}{n-n_0},
\end{equation}
the sample-weighted mean of the \emph{honest} clients' scores. Two consequences follow, and both are new. First, an adversary that repeatedly matches the previously broadcast statistic converges to Eq.~(\ref{eq:gamefix}) geometrically at rate $n_0/n<1$, so the attack is self-stabilizing rather than divergent in the score variable. Second, and more importantly, the disparity the adversary may retain while holding a zero gap is capped by the honest population's own disparity. In the limiting case where every honest client is perfectly fair, $\mathcal{F}_0^{\ast}=0$: the only way to hold a zero gap is to become perfectly fair, which is to abandon the attack. \fairfed{}'s gameability is therefore not unconditional but bounded by how unfair the honest population already is, which is the opposite of the regime in which one would deploy a fairness-aware aggregator.

\textbf{Results: the vulnerability is real but self-limiting.} The adversary must steer the statistic it is \emph{scored} on, namely the positive-label TPR gap; steering the all-sample probability gap optimizes a different quantity and never attains the condition. With the correct surrogate, the matched-gap condition is achieved on German Credit, where the adversary's gap falls from $0.271$ at round~0 to $0.0012$ averaged over the final four rounds, against $0.0002$ for honest clients. There \fairfed{} does reward it exactly as Theorem~\ref{thm:ng} predicts: the adversary's weight rises from $0.342$ to $0.349$, above the equal share of $1/K=0.333$, leaving each honest client $0.325$.

The reward, however, is nearly worthless, and Eq.~(\ref{eq:gamefix}) says why. German's honest clients report $\mathcal{F}_k=0$, so the fixed point is $\mathcal{F}_0^{\ast}=0$ and the adversary converges to $\mathcal{F}_0\approx0$: it buys $+0.016$ of aggregation weight by surrendering almost all of its bias, retaining an EOD of $0.145$ against the $0.696$ it reaches under the disparity-maximizing objective, or $21\,\%$ of its achievable disparity. On Taiwan and Adult, where honest clients are themselves unfair, the adversary retains far more disparity ($65\,\%$ and $59\,\%$) but never attains the matched-gap condition within $16$ rounds (final gaps $0.189$ and $0.102$), so \fairfed{} is not exploited there either.

The honest conclusion is therefore neither that \fairfed{} is broken nor that Theorem~\ref{thm:ng} is vacuous. The gamesmanship mechanism is empirically real, and it is bounded by Eq.~(\ref{eq:gamefix}): an adversary can convert a zero fairness gap into extra weight only in proportion to the disparity the honest population already exhibits, and in the fair-population regime that conversion rate goes to zero. \fairis{} is structurally immune because it never references $\Fg$, but under this particular attack that immunity buys little, since the attack is weak against every rule. Where \fairis{} does register a change is instructive rather than adverse: on German its adversary weight \emph{rises} from $0.218$ to $0.333$, because the adversary has genuinely become fair ($\mathcal{F}_0\approx0$) and Monotone Weight Reduction rewards low reported bias by construction (Corollary~\ref{cor:opt}). An adversary that games \fairis{} has stopped attacking. We record the adversary's aggregation weight in every round rather than only at convergence, since the predicted failure is a divergence over time rather than a single-round effect.

\section{Discussion}\label{sec:discussion}
We now turn from what \fairis{} guarantees to how to use it in practice: how to set the inclusiveness parameter $\eta$ for a given threat environment, and the limitations that bound the scheme's claims.

\subsection{Selecting $\eta$ and deployment considerations}
The choice of $\eta$ trades security against inclusiveness: $\eta=1.01$ sharply penalizes biased clients and limits the attacker to between $0.005$ and $0.165$ of the weight (Table~\ref{tab:attack}), while a larger $\eta$ such as $1.32$ preserves more equal participation when the concern is unintentional bias; Eq.~(\ref{eq:eta}) gives a default. \fairis{} is most effective when fairness varies across clients; when all are equally fair or biased the weights converge to the uniform share $1/K$, so it cannot manufacture fairness absent from every client and is not a substitute for local debiasing. The base rule also weights by fairness alone and ignores client data size $n_k$, so a data-light attacker is penalized no further than its score warrants; the size-weighted variant \fairisn{} of Section~\ref{sec:fairisn} combines the two signals, Proposition~\ref{prop:fairisn} shows every guarantee carries over, and Section~\ref{sec:standard} evaluates it. PrivFairFL~\cite{pentyala2022} targets group fairness under privacy constraints and is related but differs from \fairis{}'s provably ungameable closed-form EOD rule. FedGA~\cite{liu2025fedga} is adjacent rather than comparable: its Gini-coefficient criterion measures performance disparity \emph{across clients}, a client-level notion of fairness, not disparity across demographic groups, so it optimizes a different objective from the one studied here.

\subsection{Limitations}\label{sec:limitations}
\fairis{} has several limitations. (i)~The fairness score is sent in plaintext, so the server learns the relative bias ordering; secure computation of the weights is future work. (ii)~The properties assume honest score reporting. This is the sharpest limitation, and it is not discharged by authenticating the score channel: authentication establishes who sent a score, not that the score was computed correctly on the sender's data. Closing it requires either a proof of correct computation, such as a zero-knowledge proof over the score, or a verification protocol that resamples and checks reported scores; both are ongoing work. (iii)~\fairis{} weights on \emph{local} scores, and local fairness does not imply population fairness. Because EOD is a ratio of group- and label-conditional counts, it does not decompose into a client-size-weighted average of local EODs: two clients can each report $\Fk=0$ while the pooled model has EOD close to $1$, if their group-conditional positive rates are opposed and their group sizes are complementary. Every client can therefore look perfectly fair to the server while the population model is highly unfair, and no purely local-score rule, \fairis{} or \fairfed{}, can detect this. (iv)~The guarantees bound an adversary's \emph{influence}, not honest clients' \emph{downstream} EOD (Remark~\ref{rem:scope}): on imbalanced data such as Taiwan Credit, down-weighting the attacker does not by itself lower honest-client EOD. (v)~The empirical comparison covers \fedavg{}, a uniform-weighting ablation, and \fairfed{}; GLocalFair, SFFL, and FairTrade are discussed but not re-implemented, as their published settings differ. (vi)~The poisoning experiments use a single attacker; collusion ($|\mathcal{B}|<K/2$) is proved (Corollary~2) but not evaluated empirically. (vii)~Results use three seeds with wide deviations, and a formal convergence rate under round-varying weights remains open. Future work also includes larger-scale evaluation with more seeds and formal statistical analysis, broader fairness benchmarks such as COMPAS and folktables ACSIncome, metrics beyond EOD and SPD (equalized odds, worst-group accuracy), GLocalFair/SFFL baselines, and extension to multi-valued, intersectional, image, or text domains.

\section{Conclusion}\label{sec:conclusion}
We presented \fairis{}, a server-side aggregation-reweighting scheme that improves per-client group fairness under heterogeneous data while provably limiting the aggregation influence of high-bias participants, under honest (authenticated) fairness-score reporting, through the closed-form weight $\omega_k\propto\eta-\mathcal{F}_k$ (normalized across clients). This weight is monotone-decreasing in the attacker's bias (MWR), positive for every client (DP), and ungameable by tracking the global average (NG); MWR also extends to colluding minority coalitions (Corollary~\ref{cor:coalition}). Together these give a \emph{containment guarantee}: a more-biased adversary, alone or in coalition, always loses aggregation influence in proportion to its bias, while every client keeps positive weight, unlike \fairfed{}, whose gap rule can drive weights toward zero. This is a bound on \emph{adversarial influence}, not a promise of better downstream fairness in every regime, which additionally requires honest clients to carry sufficient fairness signal, as our three-dataset experiments show. Two distinct sets of conditions should therefore be kept separate. \emph{Influence containment} (the algebraic guarantee of Theorems~\ref{thm:mwr}--\ref{thm:ng} and Corollary~\ref{cor:coalition}) requires only the algebraic assumptions: finitely many clients with $K\geq2$, $\eta>1$, scores in $[0,1]$, honestly reported and authenticated scores, the other clients' scores held fixed for monotonicity, and a nonempty non-coalition set ($|\mathcal{B}|<K/2$) for the coalition bound. \emph{Downstream honest-client EOD improvement} is a strictly stronger, data-dependent property that additionally requires the honest clients to collectively carry sufficient fairness signal, as our three-dataset experiments show.

Future work includes secure multi-party computation of fairness scores (so the server learns only the aggregate weight and cannot infer individual clients' demographic compositions), formal characterization of the biased-client contamination threshold (the minimum fraction of fair clients required for collaboration to help rather than harm fairness), and extension to multi-valued and intersectional sensitive attributes.

\textbf{Ethical considerations.} This work uses three publicly available benchmark datasets containing no direct personal identifiers beyond the demographic attributes used in fairness evaluation. No human subjects were recruited. The aim of \fairis{} is to reduce, not enable, discriminatory outcomes in automated financial decision-making; it limits adversarial amplification of bias but does not remove bias that is present in every client's data (Section~\ref{sec:discussion}). The group-fairness metrics optimized here (EOD, SPD) do not capture all societal notions of equity; satisfying them does not guarantee the absence of harm, and deployment in high-stakes lending should involve broader validation across regulatory contexts.

\section*{Funding}
This work benefited from State aid managed by the Agence Nationale de la
Recherche (ANR) and was supported by the France ANR project
\mbox{ANR-22-CE39-0002}~EQUIHID.

\section*{Acknowledgments}
We thank Aymen Boudguiga for his contributions during the earlier development
of this work. He is not responsible for the present version.

\bibliographystyle{splncs04}
\bibliography{arxiv}
\appendix
\section{Background: Collaborative Learning and Group Fairness}\label{app:background}
This appendix provides expanded background on the collaborative learning architectures, group-fairness notions, privacy-preserving techniques, and Byzantine-robustness methods referenced throughout the paper. It is intended to make the paper self-contained for readers who are specialists in one of these areas but not all of them. The material here supports but is not required for the main results in Sections~\ref{sec:threat} through~\ref{sec:experiments}.

\subsection{The fairness-privacy-robustness landscape}
Three research threads intersect in this work, and Fairis sits at their conjunction. The first is \emph{group fairness}, which seeks to ensure that automated decisions do not systematically disadvantage demographic subgroups~\cite{dwork2012fairness,hardt2016equality,mehrabi2021survey}. The second is \emph{collaborative and federated learning}, which enables joint model training without centralizing raw data~\cite{mcmahan2017,kairouz2021advances}. The third is \emph{adversarial robustness}, which studies how malicious participants can subvert a distributed training process and how defenders can resist them~\cite{blanchard2017krum,meerza2024eabfl,kasyap2025}. Most prior work addresses one or two of these threads; the central observation motivating Fairis is that group fairness and adversarial robustness are not independent objectives in collaborative learning, because a fairness-aware aggregation rule is itself an attack surface that a rational adversary can exploit. A comprehensive recent survey of fairness in federated learning~\cite{mukhtiar2025survey} catalogs dozens of distinct works at the fairness-FL intersection but does not address the adversarial robustness of the fairness mechanism itself, which is the gap Fairis fills.

\subsection{Collaborative machine learning architectures}
Collaborative machine learning encompasses several related paradigms that keep raw data local while enabling joint model training. Understanding the distinctions between these architectures clarifies why \fairis{} is designed to be architecture-agnostic.

\textbf{Federated Learning (FL).}~\cite{mcmahan2017} In standard FL, each client $C_k$ holds a private dataset $\mathcal{D}_k$ and trains a local model $\mathcal{M}_k$ on that data. At each communication round, clients transmit their full model parameters (or gradients) to an aggregation server, which computes a weighted average (typically \fedavg{}~\cite{mcmahan2017}) and broadcasts the result. Raw data, sensitive attributes, and local training details never leave the client. The primary privacy concern is that full model parameters can leak membership and attribute information through model inversion and membership inference attacks~\cite{melis2019,shokri2017}.

\textbf{Split Learning (SL).}~\cite{gupta2018,vepakomma2018split} In Split Learning, the model is partitioned at a designated \emph{cut layer}: the client runs the forward pass up to the cut layer and transmits the intermediate activations (the \emph{smashed data}) to the server, which completes the forward and backward passes. Gradients are returned to the client up to the cut layer, and the client updates its portion of the model. Split Learning was originally motivated by distributed deep learning on sensitive health data, where raw patient records cannot leave the institution that holds them~\cite{vepakomma2018split}. It reduces the parameter surface exposed to the server relative to full-model FL but introduces smashed-data leakage risks, since the transmitted activations can, in principle, be inverted to recover input features.

\textbf{SplitML.}~\cite{trivedi2026splitml} SplitML combines federated and split learning: multiple clients train local models that share a common set of \emph{top layers} (closer to the input) while retaining private \emph{bottom layers} (closer to the output). Only the shared top-layer parameters are aggregated; the personalized bottom layers, including the classification head, never leave the client. Figure~\ref{fig:arch_app} illustrates the architecture. This design reduces the exposed parameter surface relative to full-model FL while supporting client heterogeneity through the personalized bottom layers. Because SplitML was the motivating architecture for the original \fairis{} design, we use it as the primary experimental testbed, and we additionally validate \fairis{} under full-model FL (Section~\ref{sec:experiments}); by design, the same aggregation rule extends to Split Learning, though we do not evaluate that setting empirically.

\begin{figure}[ht]
  \centering
  \includegraphics[width=0.99\columnwidth]{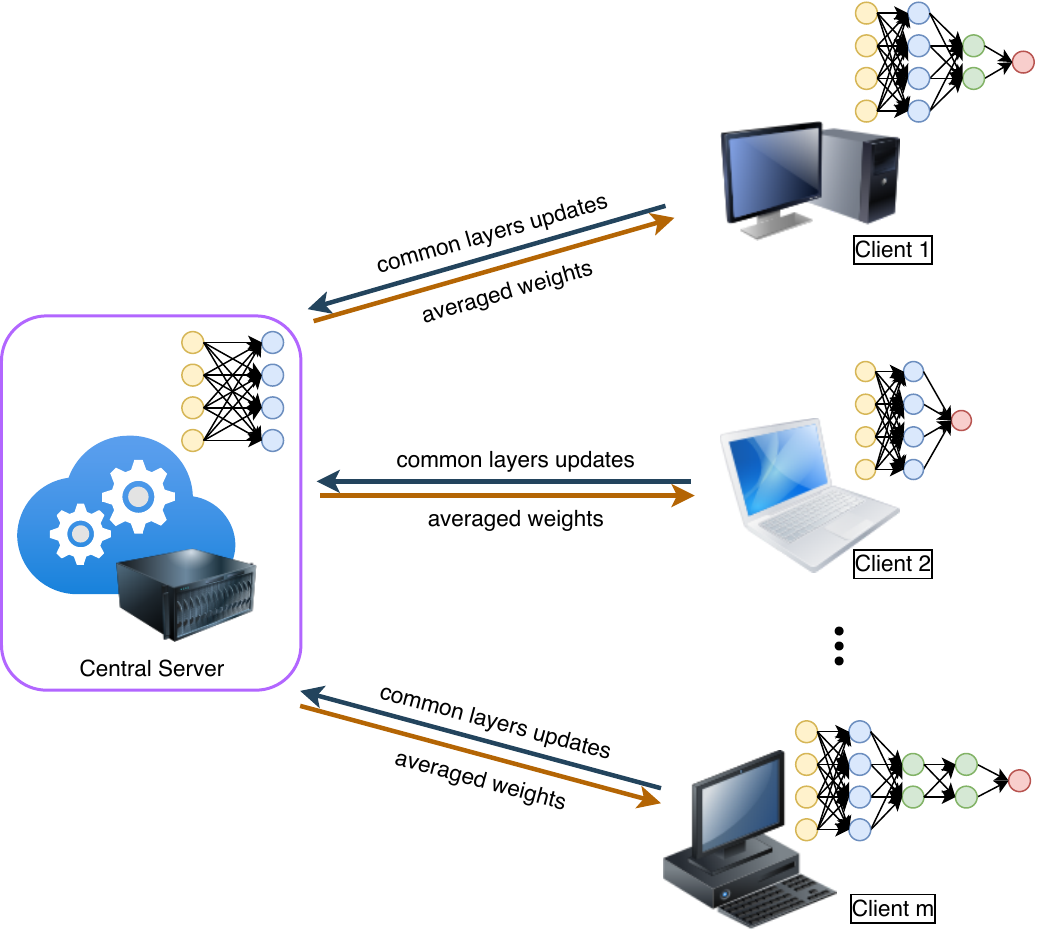}
  \caption{SplitML architecture. Each client holds a local model whose top layers (closer to the input) are common across clients and whose bottom layers (closer to the output, including the classification head) are personalized. At each round, clients send their common-layer updates to the central server, which returns the aggregated weights; the personalized bottom layers never leave the client. \fairis{} sets these aggregation weights to $\omega_k = (\eta-\Fk)/\sum_j(\eta-\mathcal{F}_j)$ in place of the size-proportional average, down-weighting biased clients while keeping every weight strictly positive. In FL mode, all layers are common, and the same weighting rule applies unchanged.}
  \label{fig:arch_app}
\end{figure}

\subsection{Group fairness metrics}
Algorithmic fairness addresses bias in automated decision-making that disproportionately harms or benefits demographic groups. We focus on two standard group-fairness notions applicable to binary classification tasks such as credit default prediction, where $Y\in\{0,1\}$ is the true label, $\hat{Y}$ is the model prediction, and $A\in\{0,1\}$ is a binary sensitive attribute ($A=1$ for the privileged group, $A=0$ for the unprivileged group). Both metrics are computed directly from the standard definitions in NumPy.

\textbf{Equal Opportunity Difference (EOD).} Equal opportunity requires that the True Positive Rate (TPR), the probability of a positive prediction given a truly positive instance, is equal across demographic groups:
\[\mathrm{EOD}=\bigl|\Pr(\hat{Y}{=}1\mid Y{=}1,A{=}1)-\Pr(\hat{Y}{=}1\mid Y{=}1,A{=}0)\bigr|.\]
EOD measures the absolute gap in TPR between the privileged and unprivileged groups. EOD=0 indicates perfect equal opportunity; EOD=1 indicates one group has zero TPR while the other has TPR=1. In credit scoring, equal opportunity means that creditworthy applicants from different demographic groups should receive credit approval at equal rates.

\textbf{Statistical Parity Difference (SPD).} Statistical parity requires that the positive prediction rate is equal across groups regardless of the true label:
\[\mathrm{SPD}=\bigl|\Pr(\hat{Y}{=}1\mid A{=}1)-\Pr(\hat{Y}{=}1\mid A{=}0)\bigr|.\]
SPD=0 indicates equal positive prediction rates (statistical parity); SPD=1 indicates one group receives positive predictions at a rate of 1 while the other receives 0. Both EOD and SPD map to $[0,1]$, with lower values indicating greater fairness. We write $\Fk$ for client $k$'s local EOD score.

\textbf{Demographic Participation (DP) and its importance.} As defined in Property~\ref{prop:dp}, DP requires that every client receive a strictly positive aggregation weight. The fairness motivation for DP is as follows: if a client is excluded (weight set to zero), its training data, including the demographic subgroup it serves, is removed from the shared model. Under non-IID data, demographic subgroups are often concentrated in specific clients. Excluding a biased client may therefore exclude the minority group that experiences discrimination, paradoxically harming the group most in need of fair treatment. \fairis{} guarantees DP by construction (Theorem~\ref{thm:dp}), whereas \fairfed{}'s gap-based update is not lower-bounded and can drive a client's weight toward zero under a sufficiently large fairness budget, so positivity is not guaranteed.

Two qualifications are warranted. First, DP is a structural rather than a numerical guarantee: a client held at $\omega_k=0.005$ contributes little to any single round, and we do not claim such a client is meaningfully represented in that round. What DP rules out is \emph{permanent elimination}, the case where a client's weight is set to exactly zero and its subgroup leaves the shared model entirely, with no path back because a zeroed client cannot influence subsequent rounds. Second, the subgroup argument presumes demographic groups are concentrated per client, which holds under strong non-IID partitioning but need not hold in general, and certainly need not hold for an adversarial client that may serve no distinct subgroup at all. DP is therefore best read as a safeguard for honest minority-serving clients rather than as a benefit conferred on the adversary.

\subsection{Gini surrogate in GLocalFair}
GLocalFair~\cite{meerza2024glocalfair} uses the Gini coefficient as a surrogate fairness measure to avoid sharing raw sensitive attribute statistics with the server. The Gini coefficient of a set of values $\{v_1,\ldots,v_n\}$ measures inequality: $G = \sum_{i,j}|v_i-v_j|/(2n\sum_i v_i)$, where $G=0$ means perfect equality and $G=1$ means maximal inequality. Applied to model prediction probabilities across groups, the Gini surrogate approximates the EOD without requiring the server to observe group labels. GLocalFair combines this surrogate with clustering-based aggregation so that clients with similar fairness profiles are grouped before averaging. The advantage is that the server learns less about individual clients' sensitive attribute distributions. The limitation is that the Gini approximation introduces noise into the fairness measurement, and the clustering step adds communication and computational overhead. \fairis{} takes a different approach: it transmits the exact EOD as a scalar, which reveals relative bias ordering across clients but not the underlying group statistics, and relies on secure multi-party computation for score privacy (identified as future work in the main paper).

\subsection{Byzantine-robust aggregation and why it does not address fairness poisoning}
Byzantine-robust aggregation rules defend federated learning against malicious clients that send arbitrary or adversarial model updates. Krum~\cite{blanchard2017krum} selects the update closest to its nearest neighbors in parameter space, discarding outliers; coordinate-wise median and trimmed-mean rules aggregate each parameter dimension robustly. These methods share a common assumption: that malicious updates are statistical outliers in parameter or gradient space, distinguishable from honest updates by distance or magnitude. Fairness poisoning attacks~\cite{meerza2024eabfl,kasyap2025} are explicitly designed to violate this assumption. The adversary maintains prediction accuracy within $\varepsilon$ of the benign baseline (Eq.~\ref{eq:goal}), so its update lies within the normal distribution of honest updates in both magnitude and predictive behavior. Only the fairness profile of the resulting model is anomalous, and fairness is not a quantity that distance-based Byzantine defenses measure. This is why \fairis{} weights by the explicit fairness score rather than by parameter-space distance: the defense must operate on the same axis along which the attack is mounted.

\subsection{Privacy-preserving techniques for collaborative learning}
Three families of privacy-preserving techniques are commonly layered onto collaborative learning, and each is compatible with \fairis{}. \emph{Secure aggregation}~\cite{bonawitz2017} uses cryptographic masking so that the server learns only the sum of client updates, not any individual update; \fairis{} weights can be computed under secure aggregation if the fairness scores are also aggregated securely. \emph{Differential privacy}~\cite{geyer2017} adds calibrated noise to client updates to bound the information any single record contributes; this is orthogonal to \fairis{}, which operates on aggregation weights rather than the updates themselves. \emph{Homomorphic encryption} and secure multi-party computation allow computation on encrypted values; the closest prior work combining fairness with these techniques is the privacy-preserving and fairness-aware framework of Bendoukha et al.~\cite{bendoukha2025popets}, which the main paper identifies as the natural next step for protecting the transmitted fairness scores. \fairis{} as presented transmits scalar fairness scores in plaintext, which reveals the relative bias ordering of clients but not the underlying group statistics; computing the weights under secure multi-party computation so the server learns only the normalized weights is left as future work.

\subsection{Complexity and convergence considerations}\label{app:complexity}
The computational overhead of \fairis{} over plain \fedavg{} is negligible. Each round, the server computes $K$ scalar subtractions ($\wbar_k = \eta - \Fk$), one sum, and $K$ divisions, for $O(K)$ additional work independent of model size; the dominant aggregation cost of the weighted parameter sum is identical to \fedavg{}. Each client additionally computes its local fairness score once per round, which requires a single forward pass over its local data already performed during evaluation. \fairis{} therefore adds no asymptotic overhead. Regarding convergence, because $\omega_k > 0$ for all clients and $\sum_k \omega_k = 1$, the \fairis{} update is a convex combination of client parameters, identical in form to \fedavg{} but with fairness-weighted coefficients. By analogy with \fedavg{} under bounded heterogeneity, convergence is expected; however, because the \fairis{} weights are functions of the model state at each round (creating a feedback loop absent in fixed-weight \fedavg{}), a formal convergence-rate proof for the round-varying case would require additional assumptions. We treat this as an open theoretical question (see the conclusion) rather than a settled claim. Characterizing the convergence rate under round-varying fairness weights and the minimum fraction of fair clients required for collaboration to improve rather than harm fairness is left for future work.

\subsection{Full SplitML and FL aggregation results}\label{app:fl_table}
Table~\ref{tab:main} gives the main SplitML EOD and SPD comparison (at $K=10$) summarized in Section~\ref{sec:experiments}; Table~\ref{tab:fl} gives the full FL full-model results (at $K=5$). All values are mean $\pm$ standard deviation over non-degenerate client-run combinations across three seeds. Table~\ref{tab:attack} reports the active fairness poisoning evaluation discussed in Section~\ref{sec:experiments}.

\begin{table}[ht]
\centering
\caption{Main comparison under Dirichlet non-IID partitioning ($\alpha=0.5$), 2 shared layers, server-side clipping enabled. Values are mean$\pm$std over all non-degenerate client-run combinations across 3 seeds. EOD and SPD are lower-is-fairer ($\downarrow$). \colorbox{fairisgreen}{\strut Shaded rows}: \fairis{}. \textbf{Bold}: best per block. Uniform ($\omega_k=1/K$) is the size-blind ablation: it discards size-proportional weighting exactly as \fairis{} does but ignores the fairness signal, so any \fairis{} result not better than Uniform is not evidence of fairness awareness.}
\label{tab:main}
\resizebox{\columnwidth}{!}{%
\begin{tabular}{llcc}
\toprule
Dataset & Method & EOD$\downarrow$ & SPD$\downarrow$ \\
\midrule
\multirow{9}{*}{German ($K{=}10$)}
  & Silo & 0.230\std{0.204} & 0.245\std{0.205} \\
  & \fedavg{} & 0.144\std{0.126} & 0.147\std{0.131} \\
  & Uniform & \textbf{0.112}\std{0.109} & \textbf{0.129}\std{0.116} \\
  & \fairfed{} $\beta{=}0$ & 0.144\std{0.126} & 0.147\std{0.131} \\
  & \fairfed{} $\beta{=}1$ & 0.144\std{0.126} & 0.147\std{0.131} \\
  & \cellcolor{fairisgreen}\fairis{} $\eta{=}1.01$ & \cellcolor{fairisgreen}\textbf{0.112}\std{0.109} & \cellcolor{fairisgreen}\textbf{0.129}\std{0.116} \\
  & \cellcolor{fairisgreen}\fairis{} $\eta{=}1.32$$^\dagger$ & \cellcolor{fairisgreen}\textbf{0.112}\std{0.109} & \cellcolor{fairisgreen}\textbf{0.129}\std{0.116} \\
  & \cellcolor{fairisgreen}\fairisn{} $\eta{=}1.01$ & \cellcolor{fairisgreen}0.144\std{0.126} & \cellcolor{fairisgreen}0.147\std{0.131} \\
  & \cellcolor{fairisgreen}\fairisn{} $\eta{=}1.32$$^\dagger$ & \cellcolor{fairisgreen}0.144\std{0.126} & \cellcolor{fairisgreen}0.147\std{0.131} \\
\midrule
\multirow{9}{*}{Taiwan ($K{=}10$)}
  & Silo & 0.401\std{0.297} & \textbf{0.373}\std{0.294} \\
  & \fedavg{} & 0.426\std{0.270} & 0.412\std{0.254} \\
  & Uniform & 0.445\std{0.273} & 0.449\std{0.267} \\
  & \fairfed{} $\beta{=}0$ & 0.426\std{0.270} & 0.412\std{0.254} \\
  & \fairfed{} $\beta{=}1$ & \textbf{0.382}\std{0.234} & 0.384\std{0.248} \\
  & \cellcolor{fairisgreen}\fairis{} $\eta{=}1.01$ & \cellcolor{fairisgreen}0.467\std{0.260} & \cellcolor{fairisgreen}0.455\std{0.266} \\
  & \cellcolor{fairisgreen}\fairis{} $\eta{=}1.32$$^\dagger$ & \cellcolor{fairisgreen}0.482\std{0.249} & \cellcolor{fairisgreen}0.473\std{0.258} \\
  & \cellcolor{fairisgreen}\fairisn{} $\eta{=}1.01$ & \cellcolor{fairisgreen}0.449\std{0.286} & \cellcolor{fairisgreen}0.435\std{0.274} \\
  & \cellcolor{fairisgreen}\fairisn{} $\eta{=}1.32$$^\dagger$ & \cellcolor{fairisgreen}0.452\std{0.280} & \cellcolor{fairisgreen}0.439\std{0.264} \\
\midrule
\multirow{9}{*}{Adult ($K{=}10$)}
  & Silo & 0.303\std{0.208} & 0.236\std{0.165} \\
  & \fedavg{} & 0.209\std{0.182} & 0.206\std{0.138} \\
  & Uniform & 0.258\std{0.170} & 0.196\std{0.128} \\
  & \fairfed{} $\beta{=}0$ & 0.209\std{0.182} & 0.206\std{0.138} \\
  & \fairfed{} $\beta{=}1$ & \textbf{0.173}\std{0.178} & \textbf{0.167}\std{0.097} \\
  & \cellcolor{fairisgreen}\fairis{} $\eta{=}1.01$ & \cellcolor{fairisgreen}0.251\std{0.193} & \cellcolor{fairisgreen}0.203\std{0.131} \\
  & \cellcolor{fairisgreen}\fairis{} $\eta{=}1.32$$^\dagger$ & \cellcolor{fairisgreen}0.254\std{0.184} & \cellcolor{fairisgreen}0.198\std{0.127} \\
  & \cellcolor{fairisgreen}\fairisn{} $\eta{=}1.01$ & \cellcolor{fairisgreen}0.202\std{0.192} & \cellcolor{fairisgreen}0.209\std{0.129} \\
  & \cellcolor{fairisgreen}\fairisn{} $\eta{=}1.32$$^\dagger$ & \cellcolor{fairisgreen}0.198\std{0.192} & \cellcolor{fairisgreen}0.205\std{0.135} \\
\bottomrule
\multicolumn{4}{l}{$^\dagger$ Exceeds recommended range $(1,(K{+}1)/K]=(1,1.1]$ for $K{=}10$; included for comparability.}
\end{tabular}%
}
\end{table}

\begin{table}[ht]
\centering
\caption{Architecture-agnostic validation: FL full-model aggregation. Values are mean$\pm$std over non-degenerate client-run combinations across 3 seeds. \colorbox{fairisgreen}{\strut Shaded}: \fairis{}. \textbf{Bold}: best per block. Uniform is the size-blind ablation as in Table~\ref{tab:main}.}
\label{tab:fl}
\resizebox{\columnwidth}{!}{%
\begin{tabular}{llcc}
\toprule
Dataset & Method & EOD$\downarrow$ & Acc \\
\midrule
\multirow{9}{*}{German ($K{=}5$)}
  & Silo & 0.272\std{0.182} & \textbf{0.621} \\
  & \fedavg{} & \textbf{0.098}\std{0.074} & 0.732 \\
  & Uniform & 0.252\std{0.176} & 0.688 \\
  & \fairfed{} $\beta{=}0$ & \textbf{0.098}\std{0.074} & 0.732 \\
  & \fairfed{} $\beta{=}1$ & \textbf{0.098}\std{0.074} & 0.732 \\
  & \cellcolor{fairisgreen}\fairis{} $\eta{=}1.01$ & \cellcolor{fairisgreen}0.252\std{0.176} & \cellcolor{fairisgreen}0.688 \\
  & \cellcolor{fairisgreen}\fairis{} $\eta{=}1.32$$^\dagger$ & \cellcolor{fairisgreen}0.252\std{0.176} & \cellcolor{fairisgreen}0.688 \\
  & \cellcolor{fairisgreen}\fairisn{} $\eta{=}1.01$ & \cellcolor{fairisgreen}\textbf{0.098}\std{0.074} & \cellcolor{fairisgreen}0.732 \\
  & \cellcolor{fairisgreen}\fairisn{} $\eta{=}1.32$$^\dagger$ & \cellcolor{fairisgreen}\textbf{0.098}\std{0.074} & \cellcolor{fairisgreen}0.732 \\
\midrule
\multirow{9}{*}{Taiwan ($K{=}5$)}
  & Silo & 0.373\std{0.266} & 0.638 \\
  & \fedavg{} & 0.161\std{0.103} & 0.733 \\
  & Uniform & 0.304\std{0.107} & 0.682 \\
  & \fairfed{} $\beta{=}0$ & 0.161\std{0.103} & 0.733 \\
  & \fairfed{} $\beta{=}1$ & 0.190\std{0.108} & \textbf{0.612} \\
  & \cellcolor{fairisgreen}\fairis{} $\eta{=}1.01$ & \cellcolor{fairisgreen}0.231\std{0.116} & \cellcolor{fairisgreen}0.702 \\
  & \cellcolor{fairisgreen}\fairis{} $\eta{=}1.32$$^\dagger$ & \cellcolor{fairisgreen}0.213\std{0.143} & \cellcolor{fairisgreen}0.709 \\
  & \cellcolor{fairisgreen}\fairisn{} $\eta{=}1.01$ & \cellcolor{fairisgreen}\textbf{0.087}\std{0.056} & \cellcolor{fairisgreen}0.719 \\
  & \cellcolor{fairisgreen}\fairisn{} $\eta{=}1.32$$^\dagger$ & \cellcolor{fairisgreen}0.105\std{0.102} & \cellcolor{fairisgreen}0.740 \\
\midrule
\multirow{9}{*}{Adult ($K{=}5$)}
  & Silo & 0.342\std{0.239} & \textbf{0.736} \\
  & \fedavg{} & 0.209\std{0.252} & 0.804 \\
  & Uniform & 0.194\std{0.084} & 0.779 \\
  & \fairfed{} $\beta{=}0$ & 0.209\std{0.252} & 0.804 \\
  & \fairfed{} $\beta{=}1$ & \textbf{0.150}\std{0.119} & 0.810 \\
  & \cellcolor{fairisgreen}\fairis{} $\eta{=}1.01$ & \cellcolor{fairisgreen}0.184\std{0.075} & \cellcolor{fairisgreen}0.783 \\
  & \cellcolor{fairisgreen}\fairis{} $\eta{=}1.32$$^\dagger$ & \cellcolor{fairisgreen}0.196\std{0.070} & \cellcolor{fairisgreen}0.781 \\
  & \cellcolor{fairisgreen}\fairisn{} $\eta{=}1.01$ & \cellcolor{fairisgreen}0.221\std{0.253} & \cellcolor{fairisgreen}0.803 \\
  & \cellcolor{fairisgreen}\fairisn{} $\eta{=}1.32$$^\dagger$ & \cellcolor{fairisgreen}0.231\std{0.246} & \cellcolor{fairisgreen}0.804 \\
\bottomrule
\multicolumn{4}{l}{$^\dagger$ Exceeds recommended range $(1,(K{+}1)/K]=(1,1.2]$ for $K{=}5$; included for comparability.}
\end{tabular}%
}
\end{table}

\begin{table}[ht]
\centering
\caption{Security evaluation under active fairness poisoning ($K=3$, $C_0$ adversarial, 2 shared layers, server-side clipping enabled, mean over 3 seeds). $\omega_0$: mean aggregation weight assigned to $C_0$. Atk.\ EOD: group disparity of $C_0$'s model. Fair EOD: mean EOD of honest clients (lower is better). Acc.\ gap: $|\text{acc}_\text{atk}-\text{acc}_\text{benign}|$ against an adversary-free control on the same partition and seed. MWR and DP are assessed per Section~\ref{sec:threat}: $\checkmark$=satisfied, $\sim$=satisfied in these runs but not guaranteed, $\times$=violated. \fairis{} is the only rule that guarantees both properties by construction; Uniform ignores the fairness signal entirely and so fixes $\omega_0=1/K$ regardless of injected bias. \fairfed{} $\beta=0$ equals \fedavg{}.}
\label{tab:attack}
\resizebox{\columnwidth}{!}{%
\begin{tabular}{llccccc c}
\toprule
Dataset & Method & $\omega_0$ & Atk.\ EOD & Fair EOD$\downarrow$ & Acc.\ gap & MWR & DP \\
\midrule
\multirow{8}{*}{German}
  & \fedavg{} & 0.395 & 0.754 & 0.378 & 0.120 & $\times$ & $\checkmark$ \\
  & Uniform & 0.333 & 0.738 & 0.420 & 0.130 & $\times$ & $\checkmark$ \\
  & \fairfed{} $\beta{=}0$ & 0.395 & 0.754 & 0.378 & 0.120 & $\times$ & $\checkmark$ \\
  & \fairfed{} $\beta{=}1$ & 0.380 & 0.696 & \textbf{0.356} & 0.145 & $\times$ & $\sim$ \\
  & \cellcolor{fairisgreen}\fairis{} $\eta{=}1.01$ & \cellcolor{fairisgreen}\textbf{0.005} & \cellcolor{fairisgreen}0.547 & \cellcolor{fairisgreen}0.386 & \cellcolor{fairisgreen}0.145 & \cellcolor{fairisgreen}$\checkmark$ & \cellcolor{fairisgreen}$\checkmark$ \\
  & \cellcolor{fairisgreen}\fairis{} $\eta{=}1.32$ & \cellcolor{fairisgreen}0.108 & \cellcolor{fairisgreen}0.696 & \cellcolor{fairisgreen}0.383 & \cellcolor{fairisgreen}0.090 & \cellcolor{fairisgreen}$\checkmark$ & \cellcolor{fairisgreen}$\checkmark$ \\
  & \cellcolor{fairisgreen}\fairisn{} $\eta{=}1.01$ & \cellcolor{fairisgreen}0.007 & \cellcolor{fairisgreen}0.557 & \cellcolor{fairisgreen}0.384 & \cellcolor{fairisgreen}0.165 & \cellcolor{fairisgreen}$\checkmark$ & \cellcolor{fairisgreen}$\checkmark$ \\
  & \cellcolor{fairisgreen}\fairisn{} $\eta{=}1.32$ & \cellcolor{fairisgreen}0.150 & \cellcolor{fairisgreen}0.624 & \cellcolor{fairisgreen}0.386 & \cellcolor{fairisgreen}0.098 & \cellcolor{fairisgreen}$\checkmark$ & \cellcolor{fairisgreen}$\checkmark$ \\
\midrule
\multirow{8}{*}{Taiwan}
  & \fedavg{} & 0.210 & 0.801 & 0.485 & 0.087 & $\times$ & $\checkmark$ \\
  & Uniform & 0.333 & 0.856 & 0.542 & 0.079 & $\times$ & $\checkmark$ \\
  & \fairfed{} $\beta{=}0$ & 0.210 & 0.801 & 0.485 & 0.087 & $\times$ & $\checkmark$ \\
  & \fairfed{} $\beta{=}1$ & \textbf{0.000} & 0.571 & 0.466 & 0.128 & $\times$ & $\times$ \\
  & \cellcolor{fairisgreen}\fairis{} $\eta{=}1.01$ & \cellcolor{fairisgreen}0.048 & \cellcolor{fairisgreen}0.692 & \cellcolor{fairisgreen}0.409 & \cellcolor{fairisgreen}0.111 & \cellcolor{fairisgreen}$\checkmark$ & \cellcolor{fairisgreen}$\checkmark$ \\
  & \cellcolor{fairisgreen}\fairis{} $\eta{=}1.32$ & \cellcolor{fairisgreen}0.154 & \cellcolor{fairisgreen}0.829 & \cellcolor{fairisgreen}0.500 & \cellcolor{fairisgreen}0.101 & \cellcolor{fairisgreen}$\checkmark$ & \cellcolor{fairisgreen}$\checkmark$ \\
  & \cellcolor{fairisgreen}\fairisn{} $\eta{=}1.01$ & \cellcolor{fairisgreen}0.007 & \cellcolor{fairisgreen}0.729 & \cellcolor{fairisgreen}\textbf{0.382} & \cellcolor{fairisgreen}0.112 & \cellcolor{fairisgreen}$\checkmark$ & \cellcolor{fairisgreen}$\checkmark$ \\
  & \cellcolor{fairisgreen}\fairisn{} $\eta{=}1.32$ & \cellcolor{fairisgreen}0.078 & \cellcolor{fairisgreen}0.776 & \cellcolor{fairisgreen}0.476 & \cellcolor{fairisgreen}0.109 & \cellcolor{fairisgreen}$\checkmark$ & \cellcolor{fairisgreen}$\checkmark$ \\
\midrule
\multirow{8}{*}{Adult}
  & \fedavg{} & 0.303 & 0.858 & 0.519 & 0.236 & $\times$ & $\checkmark$ \\
  & Uniform & 0.333 & 0.719 & \textbf{0.479} & 0.254 & $\times$ & $\checkmark$ \\
  & \fairfed{} $\beta{=}0$ & 0.303 & 0.858 & 0.519 & 0.236 & $\times$ & $\checkmark$ \\
  & \fairfed{} $\beta{=}1$ & 0.180 & 0.611 & 0.499 & 0.210 & $\times$ & $\sim$ \\
  & \cellcolor{fairisgreen}\fairis{} $\eta{=}1.01$ & \cellcolor{fairisgreen}0.183 & \cellcolor{fairisgreen}0.387 & \cellcolor{fairisgreen}0.510 & \cellcolor{fairisgreen}0.201 & \cellcolor{fairisgreen}$\checkmark$ & \cellcolor{fairisgreen}$\checkmark$ \\
  & \cellcolor{fairisgreen}\fairis{} $\eta{=}1.32$ & \cellcolor{fairisgreen}0.230 & \cellcolor{fairisgreen}0.461 & \cellcolor{fairisgreen}0.519 & \cellcolor{fairisgreen}0.238 & \cellcolor{fairisgreen}$\checkmark$ & \cellcolor{fairisgreen}$\checkmark$ \\
  & \cellcolor{fairisgreen}\fairisn{} $\eta{=}1.01$ & \cellcolor{fairisgreen}0.185 & \cellcolor{fairisgreen}0.516 & \cellcolor{fairisgreen}0.525 & \cellcolor{fairisgreen}0.183 & \cellcolor{fairisgreen}$\checkmark$ & \cellcolor{fairisgreen}$\checkmark$ \\
  & \cellcolor{fairisgreen}\fairisn{} $\eta{=}1.32$ & \cellcolor{fairisgreen}\textbf{0.146} & \cellcolor{fairisgreen}0.720 & \cellcolor{fairisgreen}0.486 & \cellcolor{fairisgreen}0.206 & \cellcolor{fairisgreen}$\checkmark$ & \cellcolor{fairisgreen}$\checkmark$ \\
\bottomrule
\end{tabular}%
}
\end{table}

\begin{figure}[ht]
  \centering
  \includegraphics[width=0.99\columnwidth]{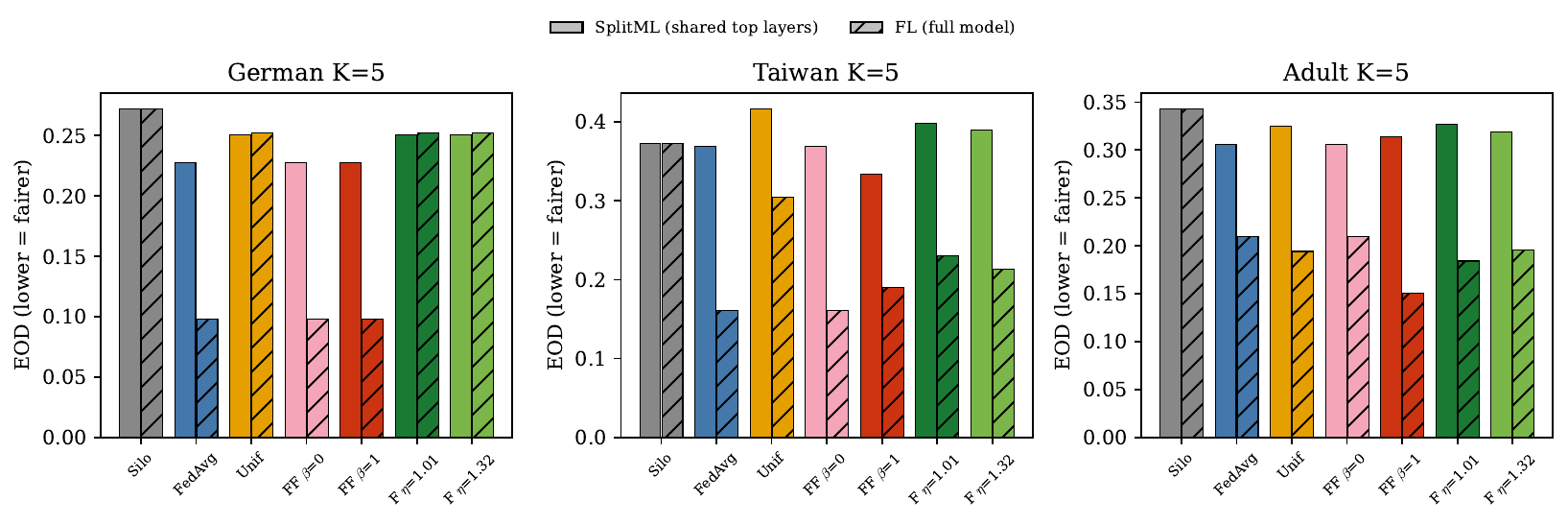}
  \caption{Architecture-agnostic transfer. Solid bars: SplitML mode (only the shared top, input-side layers are aggregated; the personalized bottom layers remain private). Hatched bars: FL mode (full model aggregated, all parameters shared). The \fairis{} weighting rule $\omega_k=(\eta-\Fk)/\sum_j(\eta-\mathcal{F}_j)$ is applied unchanged in both settings, shown for German, Taiwan, and Adult at $K=5$. The same closed-form aggregation runs on both architectures without modification while keeping every client's weight strictly positive. In the poisoning-free FL regime (Table~\ref{tab:fl}), \fedavg{} attains the lowest EOD on German and Taiwan while \fairfed{} ($\beta{=}1$) attains the lowest on Adult, where both fairness-aware rules improve on \fedavg{}; the figure's main point is the architecture independence of the mechanism.}
  \label{fig:fl_vs_splitml}
\end{figure}

\begin{figure}[ht]
  \centering
  \includegraphics[width=0.99\columnwidth]{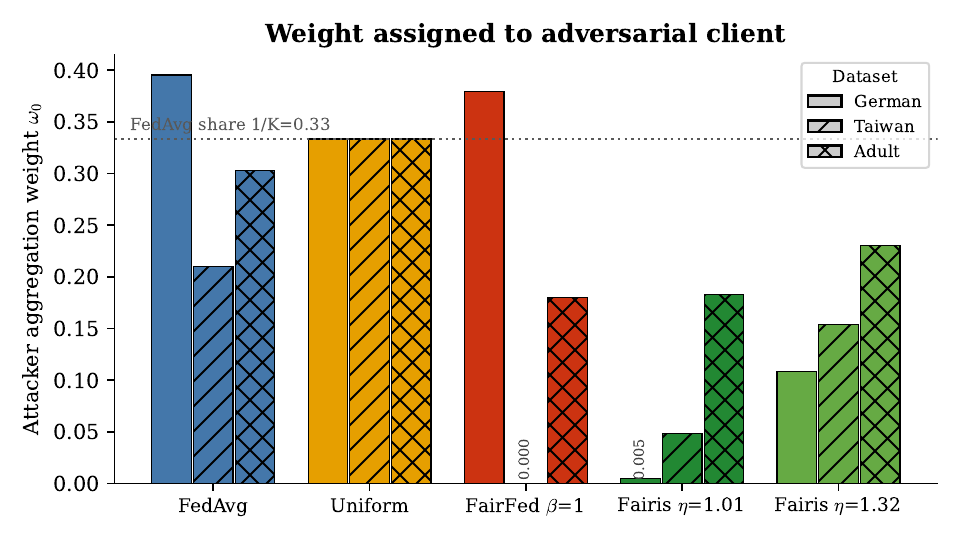}
  \caption{Empirical verification of Monotone Weight Reduction and Demographic Participation: aggregation weight $\omega_0$ assigned to the adversarial client under active fairness poisoning at $\alpha_\text{adv}=2$ ($K=3$, mean over three seeds; solid bars German Credit, hatched Taiwan Credit, cross-hatched Adult Income; dotted line marks the equal share $1/K=0.33$). \fedavg{} gives the attacker its data-proportional share ($0.395$, $0.210$, $0.303$) and Uniform fixes it at $1/K$ on every dataset, neither responding to injected bias. \fairfed{} ($\beta=1$) drives the weight to exactly $0$ on Taiwan ($0.380$, $0.000$, $0.180$), violating Demographic Participation there. Under \fairis{} ($\eta=1.01$) the weight stays strictly positive everywhere (Theorem~\ref{thm:dp}) at $0.005$, $0.048$, $0.183$, the lowest of any rule on German and Taiwan and within $0.003$ of clipped \fairfed{} on Adult; at $\eta=1.32$ it is $0.108$, $0.154$, $0.230$. Adult is the dataset where absolute-score weighting has least to act on, for the reason analyzed in Section~\ref{sec:sweep}.}
  \label{fig:mwr}
\end{figure}

\subsection{Worked example}\label{app:worked}
To make the weighting concrete, consider $K=3$ clients with local EOD scores $\mathcal{F}_0=0.85$ (a biased client), $\mathcal{F}_1=0.22$, and $\mathcal{F}_2=0.18$ (two fair clients), using $\eta=1.01$. The unnormalized weights are $\wbar_0 = 1.01-0.85 = 0.16$, $\wbar_1 = 1.01-0.22 = 0.79$, and $\wbar_2 = 1.01-0.18 = 0.83$, summing to $1.78$. The normalized weights are therefore $\omega_0 = 0.16/1.78 \approx 0.090$, $\omega_1 = 0.79/1.78 \approx 0.444$, and $\omega_2 = 0.83/1.78 \approx 0.466$. The biased client receives roughly one-fifth of the weight of either fair client, but its weight remains strictly positive, so its data still contributes to the shared model. Figure~\ref{fig:arch_app} shows the SplitML architecture to which these weights are applied. Under \fedavg{}, all three clients would receive weight $1/3 \approx 0.333$ regardless of bias; under \fairfed{} with a sufficiently large budget, the biased client's weight could be driven toward zero, removing its subgroup. \fairis{} thus interpolates between these extremes in a single closed-form step. The full weight-versus-bias relationship is plotted in Figure~\ref{fig:weight}.

\subsection{Five limitations of \fairfed{} that \fairis{} resolves}\label{app:fairfed_limits}
We identify five structural limitations of \fairfed{}~\cite{ezzeldin2023} that arise in any collaborative learning setting, including a new security limitation (L5) not previously identified in the literature. Understanding these limitations motivates the design choices in \fairis{}. (L1)~\fairfed{} optimizes the global model's fairness but does not guarantee per-client fairness improvement; a client with an already-biased dataset may see its local model become less fair after aggregation. (L2)~\fairfed{} is designed to work alongside optional local pre-processing debiasing on each client; even if clients apply such debiasing, the collaborative aggregation can undo those local fairness gains, since clients with more biased data contribute updates that dominate the aggregate and shift the global model back toward bias. (L3)~\fairfed{}'s aggregation weight update rule is not lower-bounded and can produce negative weights, which the standard implementation clips to zero; under a sufficiently large fairness budget a biased client can therefore lose all weight, removing its demographic subgroup from the shared representation and potentially harming the very group that fairness interventions are meant to protect (Property~DP not guaranteed). (L4)~\fairfed{} minimizes the \emph{gap} between a client's local fairness score and the global average, so a client with a small gap but absolutely high bias receives more weight than a fairer client with a larger gap; this can reinforce systemic bias. (L5, new)~\fairfed{}'s gap depends on the global fairness score, which an adversary can infer from the broadcast model; a $\PPT$ adversary can therefore set $\Fk\approx\Fg$, minimizing its weight penalty while maintaining positive bias (Property~NG violated, proved in Theorem~\ref{thm:ng}). \fairis{} addresses all five limitations by using the absolute score in a closed-form positive-weight formula.

\subsection{Convergence of global-model fairness}\label{app:convergence}
Figure~\ref{fig:convergence} tracks the global-model EOD over communication rounds on German Credit (full-model FL, $K=5$, three-round moving average over three seeds). On this small dataset the local fairness-score channel saturates: every client interpolates its small training split, so the reported scores are near zero, \fairfed{} collapses onto \fedavg{} at both $\beta$ values (the two curves coincide exactly), and \fairis{} reduces to near-uniform averaging. The curves stay separated throughout: \fedavg{} (equivalently \fairfed{}) fluctuates between roughly $0.10$ and $0.14$, while \fairis{} ($\eta=1.01$) declines from about $0.28$ toward $0.22$ and remains above it, consistent with the full-model $K=5$ result in Table~\ref{tab:fl} (\fedavg{} $0.098$, \fairis{} $0.252$). We include this trace only as a convergence illustration on a saturated boundary case, not as a fairness ranking; the method's fairness behavior is characterized on the larger datasets and, most importantly, under active poisoning (Section~\ref{sec:security}).

\begin{figure}[ht]
  \centering
  \includegraphics[width=0.99\columnwidth]{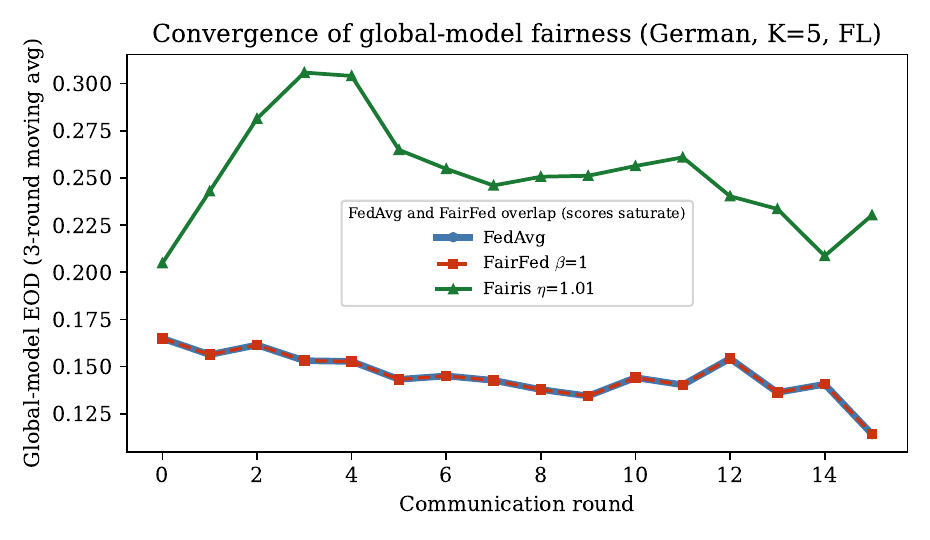}
  \caption{Global-model fairness over communication rounds (German Credit, full-model FL, $K=5$, three-round moving average of EOD over three seeds). On this small dataset the score channel saturates, so \fairfed{} ($\beta{=}1$) coincides exactly with \fedavg{} (overlapping curves) and \fairis{} reduces to near-uniform averaging; \fedavg{} holds the lower EOD throughout and \fairis{} ($\eta=1.01$) remains above it, consistent with Table~\ref{tab:fl}.}
  \label{fig:convergence}
\end{figure}

\subsection{Proofs of the security properties}\label{app:proofs}
This appendix gives the full proofs of the three security properties and the worst-case weight bound stated in Section~\ref{sec:security}.

\begin{proof}[of Theorem~\ref{thm:mwr}, MWR]
Hold every other client's score $\{\mathcal{F}_j\}_{j\ne k}$ fixed and define $S_{-k}=\sum_{j\ne k}(\eta-\mathcal{F}_j)$. Since $\eta>1$ and $\mathcal{F}_j\in[0,1]$, each term satisfies $\eta-\mathcal{F}_j\geq\eta-1>0$; with $K\geq2$ the index set $\{j\ne k\}$ is non-empty, so $S_{-k}>0$. (For $K=1$ the sum is empty, $S_{-k}=0$, and $\omega_k\equiv1$ is constant.) Writing $\omega_k=(\eta-\Fk)/(S_{-k}+\eta-\Fk)$ and differentiating with respect to $\Fk$:
\[\frac{\partial\omega_k}{\partial\Fk} = \frac{-S_{-k}}{(S_{-k}+\eta-\Fk)^2} < 0,\]
since $S_{-k}>0$ makes the numerator $-S_{-k}$ strictly negative while the squared denominator is strictly positive. Therefore $\omega_k$ is strictly decreasing in $\Fk$.
\end{proof}

\begin{proof}[of Corollary~\ref{cor:opt}]
By Theorem~\ref{thm:mwr}, $\omega_k$ is strictly decreasing in $\Fk$. The maximum is therefore achieved at $\Fk=0$, where $\omega_k=\eta/(\eta + S_{-k})$. Any $\Fk>0$ yields $\omega_k<\eta/(\eta+S_{-k})$.
\end{proof}

\begin{proof}[of Corollary~\ref{cor:coalition}, Coalition weight bound]
Let $W_\mathcal{B}=\sum_{k\in\mathcal{B}}(\eta-\mathcal{F}_k)$ and $W_{\neg\mathcal{B}}=\sum_{k\notin\mathcal{B}}(\eta-\mathcal{F}_k)$. Since $\eta>1$ and $\mathcal{F}_j\in[0,1]$, each outside term satisfies $\eta-\mathcal{F}_j\geq\eta-1>0$. The number of non-coalition clients is $K-|\mathcal{B}|$, which is at least $1$ because $|\mathcal{B}|<K/2$ by (A1); hence $W_{\neg\mathcal{B}}\geq(K-|\mathcal{B}|)(\eta-1)>0$. The argument uses only the positivity $W_{\neg\mathcal{B}}>0$, so the exact lower bound is immaterial; note the bound depends on the count of \emph{non-coalition} clients $K-|\mathcal{B}|$, not on $|\mathcal{B}|$. The total coalition weight is $\sum_{k\in\mathcal{B}}\omega_k = W_\mathcal{B}/(W_\mathcal{B}+W_{\neg\mathcal{B}})$. Differentiating with respect to any $\mathcal{F}_{k'}$ for $k'\in\mathcal{B}$:
\[\frac{\partial}{\partial\mathcal{F}_{k'}}\frac{W_\mathcal{B}}{W_\mathcal{B}+W_{\neg\mathcal{B}}} = \frac{-1\cdot(W_\mathcal{B}+W_{\neg\mathcal{B}})-W_\mathcal{B}\cdot(-1)}{(W_\mathcal{B}+W_{\neg\mathcal{B}})^2} = \frac{-W_{\neg\mathcal{B}}}{(W_\mathcal{B}+W_{\neg\mathcal{B}})^2} < 0.\]
The coalition's total weight is therefore maximized at $\mathcal{F}_k=0$ for all $k\in\mathcal{B}$, at which point the attack has zero bias impact per Eq.~(\ref{eq:goal}).
\end{proof}

\begin{proof}[of Theorem~\ref{thm:dp}, DP]
For any $\Fk\in[0,1]$, we have $\wbar_k=\eta-\Fk\geq\eta-1>0$ (since $\eta>1$). Because $\wbar_k>0$ and $\wbar_j>0$ for all $j$, the normalized weight $\omega_k=\wbar_k/\sum_j\wbar_j>0$. No client is ever assigned zero or negative weight, so no client's demographic subgroup is excluded from the aggregated model.
\end{proof}

\begin{proof}[of Theorem~\ref{thm:ng}, NG]
\emph{Under \fairis{}:} By Theorem~\ref{thm:mwr}, $\omega_k$ is strictly decreasing in $\Fk$. A perfectly fair client achieves weight $\eta/(\eta+S_{-k})$. Any $\Fk>0$ yields strictly lower weight. There is therefore no strategy that maintains $\Fk>0$ and achieves the weight corresponding to $\Fk=0$.

\emph{Under \fairfed{}:} We analyze the \emph{unclipped} \fairfed{} weight update $\wbar_k^t=\wbar_k^{t-1}-\beta(\Delta_k - \bar\Delta)$, where $\Delta_k=|\Fg-\Fk|$ and $\bar\Delta=(1/K)\sum_i\Delta_i$, under three explicit premises: a positive fairness budget $\beta>0$, a positive total unnormalized mass $\sum_j\wbar_j^{t-1}>0$, and at least one honest client with a nonzero gap, so $\bar\Delta>0$ (if all gaps vanish no client is reweighted and the update is the identity). On timing, we treat $\Fg$ as the scalar broadcast at the end of the previous round: the adversary observes it (it is inferable from the broadcast model) and matches it, so $\Fg$ is a fixed constant when the adversary chooses $\Fk$. Setting $\Fk=\Fg$ gives $\Delta_k=0$, and the update becomes $\wbar_k^t=\wbar_k^{t-1}-\beta(0-\bar\Delta)=\wbar_k^{t-1}+\beta\bar\Delta$, which strictly exceeds $\wbar_k^{t-1}$ because $\beta\bar\Delta>0$; the adversary's unnormalized weight therefore grows. For honest clients the update $\wbar_j^t=\wbar_j^{t-1}-\beta(\Delta_j-\bar\Delta)$ reduces their unnormalized weight whenever $\Delta_j>\bar\Delta$. Summing over clients, $\sum_j\wbar_j^t=\sum_j\wbar_j^{t-1}-\beta\sum_j(\Delta_j-\bar\Delta)=\sum_j\wbar_j^{t-1}$, since $\sum_j(\Delta_j-\bar\Delta)=0$ by definition of $\bar\Delta$. The total mass is thus conserved and remains positive, so the adversary's normalized weight $\omega_0^t=\wbar_0^t/\sum_j\wbar_j^t$ strictly increases as $\wbar_0^t$ grows. The adversary's weight grows over rounds rather than decreasing, despite maintaining $\Fk=\Fg>0$, which proves the gamesmanship vulnerability. The formal argument above analyzes the unclipped gap update. The clipped implementation that floors weights at zero alters the normalization denominator, so the algebraic conclusion does not transfer verbatim. It does, however, transfer in the direction that matters, because the floor can bind only on clients whose unnormalized weight is being driven down, never on the matched-gap adversary whose weight is increasing. Iterating the clipped update confirms this: with $K=3$ equal-sized clients, honest scores $0.10$ and $0.85$, $\beta=1$, and an adversary that matches the broadcast $\Fg$ each round, the adversary's normalized share rises $0.550\to0.789\to0.982\to1.000$ over four rounds, at which point both honest clients have been floored to zero and the adversary holds the entire aggregate. Clipping therefore converts the gamesmanship vulnerability from unbounded weight growth into total capture, and additionally violates DP for the \emph{honest} clients.
\end{proof}

\subsection{Novelty comparison table}\label{app:novelty}
Table~\ref{tab:novelty} summarizes how \fairis{} compares with the most closely related methods across the five dimensions discussed in Section~\ref{sec:related}.

\begin{table}[ht]
\centering
\caption{Novelty comparison across five security and fairness dimensions. GF: group fairness target; PC: per-client fairness scope; AI: architecture-independent (weighting rule validated on two or more distinct aggregation architectures); PW: guaranteed positive weight for every client; AR: formal resistance to fairness poisoning attacks. $\checkmark$/--/(atk): satisfied/not addressed/attack paper that identifies the gap.}
\label{tab:novelty}
\begin{tabular}{lccccc}
\toprule
Method & GF & PC & AI & PW & AR \\
\midrule
\fairfed{}~\cite{ezzeldin2023} & $\checkmark$ & -- & -- & -- & -- \\
GLocalFair~\cite{meerza2024glocalfair} & $\checkmark$ & $\checkmark$ & -- & -- & -- \\
SFFL~\cite{zhang2025sffl} & $\checkmark$ & $\checkmark$ & -- & -- & -- \\
SplitLPF~\cite{chen2024splitlpf} & -- & $\checkmark$ & -- & -- & -- \\
EAB-FL~\cite{meerza2024eabfl} & $\checkmark$ & -- & -- & -- & (atk) \\
Kasyap et al.~\cite{kasyap2025} & $\checkmark$ & -- & -- & -- & (atk) \\
\midrule
\rowcolor{fairisgreen}\fairis{} (ours) & $\checkmark$ & $\checkmark$ & $\checkmark$ & $\checkmark$ & $\checkmark$ \\
\bottomrule
\end{tabular}
\end{table}

\subsection{SplitML EOD by client count}\label{app:eod_fig}
Figure~\ref{fig:eod} shows the per-method SplitML EOD across client counts $K$ for German Credit, Taiwan Credit, and Adult Income, complementing the representative-$K$ values in Table~\ref{tab:main}.

\begin{figure}[ht]
  \centering
  \includegraphics[width=0.99\columnwidth]{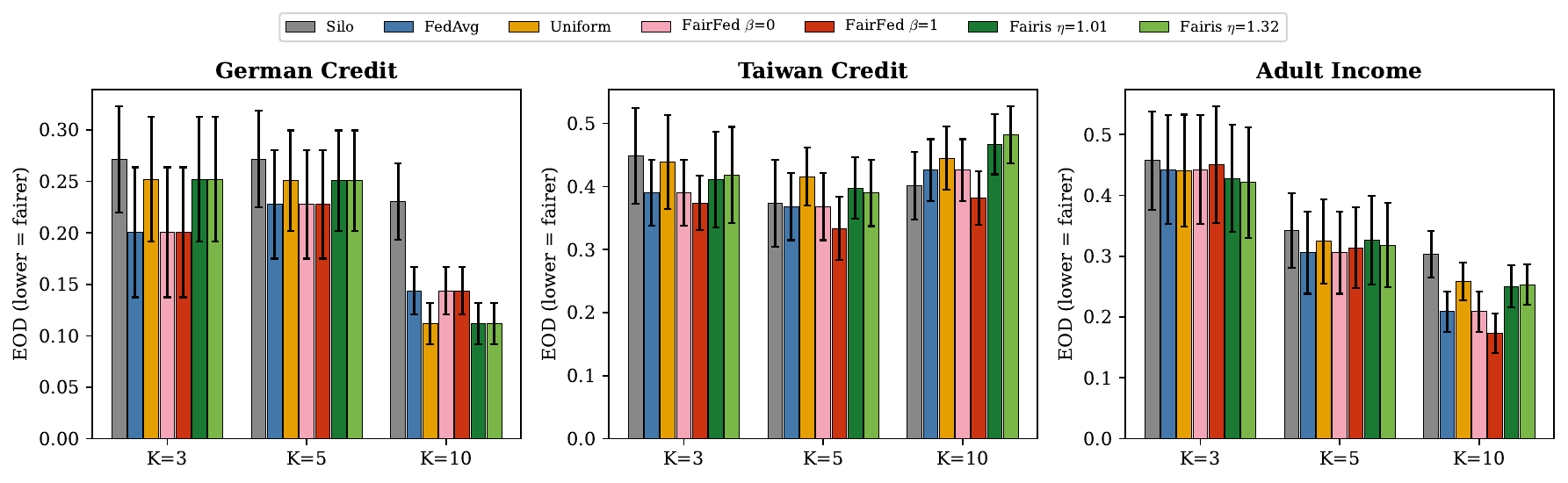}
  \caption{Mean Equal Opportunity Difference (EOD, lower is fairer) per method and number of clients $K$ under Dirichlet non-IID partitioning ($\alpha=0.5$), SplitML mode with 2 shared layers, averaged over non-degenerate client-run combinations from 3 seeds. Error bars are the standard error of the mean. No single method dominates across all datasets and client counts in this routine non-IID setting; \fairis{} keeps every client's weight strictly positive throughout, and the method's distinctive behavior appears under active poisoning (Table~\ref{tab:attack}) rather than in the poisoning-free comparison shown here.}
  \label{fig:eod}
\end{figure}

\end{document}